\PassOptionsToPackage{unicode}{hyperref}
\PassOptionsToPackage{hyphens}{url}
\PassOptionsToPackage{dvipsnames,svgnames,x11names}{xcolor}
\documentclass[
  12pt]{article}
\usepackage{xcolor}
\usepackage{amsmath,amssymb}
\usepackage{iftex}
\ifPDFTeX
  \usepackage[T1]{fontenc}
  \usepackage[utf8]{inputenc}
  \usepackage{textcomp} 
\else 
  \usepackage{unicode-math} 
  \defaultfontfeatures{Scale=MatchLowercase}
  \defaultfontfeatures[\rmfamily]{Ligatures=TeX,Scale=1}
\fi
\usepackage{lmodern}
\ifPDFTeX\else
\fi
\IfFileExists{upquote.sty}{\usepackage{upquote}}{}
\IfFileExists{microtype.sty}{
  \usepackage[]{microtype}
  \UseMicrotypeSet[protrusion]{basicmath} 
}{}
\makeatletter
\@ifundefined{KOMAClassName}{
  \IfFileExists{parskip.sty}{%
    \usepackage{parskip}
  }{
    \setlength{\parindent}{0pt}
    \setlength{\parskip}{6pt plus 2pt minus 1pt}}
}{
  \KOMAoptions{parskip=half}}
\makeatother
\makeatletter
\ifx\paragraph\undefined\else
  \let\oldparagraph\paragraph
  \renewcommand{\paragraph}{
    \@ifstar
      \xxxParagraphStar
      \xxxParagraphNoStar
  }
  \newcommand{\xxxParagraphStar}[1]{\oldparagraph*{#1}\mbox{}}
  \newcommand{\xxxParagraphNoStar}[1]{\oldparagraph{#1}\mbox{}}
\fi
\ifx\subparagraph\undefined\else
  \let\oldsubparagraph\subparagraph
  \renewcommand{\subparagraph}{
    \@ifstar
      \xxxSubParagraphStar
      \xxxSubParagraphNoStar
  }
  \newcommand{\xxxSubParagraphStar}[1]{\oldsubparagraph*{#1}\mbox{}}
  \newcommand{\xxxSubParagraphNoStar}[1]{\oldsubparagraph{#1}\mbox{}}
\fi
\makeatother

\usepackage{longtable,booktabs,array}
\usepackage{calc} 
\usepackage{etoolbox}
\makeatletter
\patchcmd\longtable{\par}{\if@noskipsec\mbox{}\fi\par}{}{}
\makeatother
\IfFileExists{footnotehyper.sty}{\usepackage{footnotehyper}}{\usepackage{footnote}}
\makesavenoteenv{longtable}
\usepackage{graphicx}
\makeatletter
\newsavebox\pandoc@box
\newcommand*\pandocbounded[1]{
  \sbox\pandoc@box{#1}%
  \Gscale@div\@tempa{\textheight}{\dimexpr\ht\pandoc@box+\dp\pandoc@box\relax}%
  \Gscale@div\@tempb{\linewidth}{\wd\pandoc@box}%
  \ifdim\@tempb\p@<\@tempa\p@\let\@tempa\@tempb\fi
  \ifdim\@tempa\p@<\p@\scalebox{\@tempa}{\usebox\pandoc@box}%
  \else\usebox{\pandoc@box}%
  \fi%
}
\def\fps@figure{htbp}
\makeatother

\providecommand{\tightlist}{%
  \setlength{\itemsep}{0pt}\setlength{\parskip}{0pt}}

\usepackage[]{natbib}
\usepackage[no-weekday]{eukdate}
\usepackage{mathpazo,mathtools,todonotes,bm, dsfont}
\usepackage{caption}
\mathtoolsset{mathic=true}

\newcommand{\E}{\mathbb{E}}

\AtBeginDocument{\urlstyle{tt}\ifdefined\harvardurl\renewcommand{\harvardurl}[1]{\textbf{URL:} \url{#1}}\fi}
\makeatletter
\@ifpackageloaded{caption}{}{\usepackage{caption}}
\AtBeginDocument{%
\ifdefined\contentsname
  \renewcommand*\contentsname{Table of contents}
\else
  \newcommand\contentsname{Table of contents}
\fi
\ifdefined\listfigurename
  \renewcommand*\listfigurename{List of Figures}
\else
  \newcommand\listfigurename{List of Figures}
\fi
\ifdefined\listtablename
  \renewcommand*\listtablename{List of Tables}
\else
  \newcommand\listtablename{List of Tables}
\fi
\ifdefined\figurename
  \renewcommand*\figurename{Figure}
\else
  \newcommand\figurename{Figure}
\fi
\ifdefined\tablename
  \renewcommand*\tablename{Table}
\else
  \newcommand\tablename{Table}
\fi
}
\@ifpackageloaded{float}{}{\usepackage{float}}
\floatstyle{ruled}
\@ifundefined{c@chapter}{\newfloat{codelisting}{h}{lop}}{\newfloat{codelisting}{h}{lop}[chapter]}
\floatname{codelisting}{Listing}

\usepackage{amsthm}
\theoremstyle{plain}
\newtheorem{theorem}{Theorem}[section]
\theoremstyle{definition}
\newtheorem{definition}{Assumption}[section]
\theoremstyle{plain}
\newtheorem{lemma}{Lemma}[section]
\theoremstyle{remark}
\AtBeginDocument{}

\makeatother
\makeatletter
\@ifpackageloaded{caption}{}{\usepackage{caption}}
\@ifpackageloaded{subcaption}{}{\usepackage{subcaption}}
\makeatother
\usepackage{bookmark}
\IfFileExists{xurl.sty}{\usepackage{xurl}}{} 
\makeatletter
\@ifundefined{xmpquote}{\newcommand{\xmpquote}[1]{#1}}{}
\makeatother
\hypersetup{
  pdftitle={Anomaly detection using surprisals},
  pdfauthor={Rob J Hyndman; David T Frazier},
  pdfkeywords={\xmpquote{Extreme value theory}, \xmpquote{Generalized
Pareto Distribution}, \xmpquote{Outlier detection}, \xmpquote{Tail
bounds}},
  colorlinks=true,
  linkcolor={blue},
  filecolor={Maroon},
  citecolor={Blue},
  urlcolor={Blue},
  pdfcreator={LaTeX via pandoc}}

\begin{document}

\def\spacingset#1{\renewcommand{\baselinestretch}%
{#1}\small\normalsize} \spacingset{1}


\title{\bf Anomaly detection using surprisals}
\author{
Rob J Hyndman\\
Department of Econometrics \& Business Statistics, Monash University\\
and\\David T Frazier\\
Department of Econometrics \& Business Statistics, Monash University\\
}
\maketitle

\bigskip
\bigskip
\begin{abstract}
Anomaly detection methods are widely used but often rely on ad hoc rules
or strong assumptions, and they often focus on tail events, missing
``inlier'' anomalies that occur in low-density gaps between modes. We
propose a unified framework that defines an anomaly as an observation
with unusually low probability under a (possibly mis-specified) model.
For each observation we compute its surprisal (the negative log
generalized density) and define an anomaly score as the probability of a
surprisal at least as large as that observed. This reduces anomaly
detection for complex univariate or multivariate data to estimating the
upper tail of a univariate surprisal distribution. We develop two
model-robust approximations to these tail probabilities: an empirical
approximation based on the observed surprisal distribution and a GPD
approximation that fits a Generalized Pareto Distribution above a high
threshold. For the empirical approximation we give conditions under
which tail ordering is preserved and derive finite-sample confidence
guarantees via the Dvoretzky--Kiefer--Wolfowitz inequality. For the GPD
approximation we establish broad tail conditions ensuring classical
extreme-value behavior. Simulations and applications to French
mortality, Test-cricket and Old Faithful geyser data show the approach
remains effective under substantial model mis-specification.
\end{abstract}

\noindent%
\vfill

\newpage
\spacingset{1.9} 

\section{Introduction}\label{introduction}

Anomaly detection methods are often rather ad hoc and lack a solid
theoretical foundation, or they make strong distributional assumptions
that may not hold in practice. We aim to address this by defining
anomalies as observations of low probability under a specified model,
and proposing a general framework for anomaly detection that estimates
the probability mass associated with observations of equal or lower
density (thus more anomalous). The approach applies to any anomaly
detection problem for which a probability distribution can be defined on
the space of observations, including univariate and multivariate data,
and unconditional and conditional distributions. Although it requires a
specified probability distribution, it is robust to substantial
mis-specification of that distribution.

The basic idea is to calculate the surprisal (equal to the negative log
density) of each observation, and then estimate the tail probability
from the surprisal distribution. This gives the probability of obtaining
a surprisal at least as large as that of the observation, and a
threshold on these probabilities controls the false positive rate when
flagging anomalies.

Let \(\bm{y}_1,\dots,\bm{y}_n\) be iid observations from a probability
distribution \(F\) defined on an appropriate sample space, and let \(f\)
be the generalized probability density function of \(F\). That is, \(f\)
is defined using Dirac delta functions at jump points of \(F\), so that
it represents the density function if \(F\) is continuous, the
probability mass function if \(F\) is discrete, and a combination of
both if \(F\) is a mixture of continuous and discrete components.

The \emph{surprisal} of observation \(\bm{y}_i\) is defined as
\citep{Stone2022}
\begin{equation}\protect\phantomsection\label{eq-surprisal}{
  s_i = -\log f(\bm{y}_i).
}\end{equation} While these are better known as ``log scores'' in
statistics \citep{Good1952-zw}, in this context we prefer the term
``surprisal'' \citep[coined by][]{Tribus61}, as it captures the idea
that observations with low density are more ``surprising'' (and
therefore potentially anomalous) than those with high density. A large
value of \(s_i\) indicates that \(\bm{y}_i\) is an unlikely value, and
so is a potential anomaly. On the other hand, typical values will have
low surprisals. In information theory, the average surprisal is known as
the entropy of a random variable \citep{Cover2006, Stone2022} and
provides a measure of uncertainty. If \(f\) is obtained from a
parametric model, then the sum of the surprisals is equal to minus the
log likelihood of the data.

Let \(G(s) = \Pr\{S \le s\}\) be the distribution function of the
surprisal values, where \(S = -\log f(\bm{Y})\) and \(\bm{Y} \sim F\).
Then the probability of observing a value with surprisal at least as
large as \(s\) is given by \(\Pr(S \ge s) = 1 - G(s^-)\), where
\(G(s^-) = \lim_{t \uparrow s} G(t)\). This leads to a natural way of
assigning an anomaly score to each observation \(\bm{y}_i\):
\begin{equation}\protect\phantomsection\label{eq-pi}{
  p_i = \Pr(S \ge s_i) = 1 - G(s_i^-),
}\end{equation} where \(s_i\) is given by (\ref{eq-surprisal}). We
define anomalies as observations with \(p_i \le \alpha\), giving a false
positive rate of \(\alpha\) if all observations are from \(F\). This is
preferable to directly computing the probabilities from the tail of
\(F\), as it allows for anomalies to occur anywhere in the distribution
of \(F\), not just in its tails. For example, anomalies may lie between
modes in a multimodal distribution, or in otherwise low-density regions
of the distribution \(F\).

Distance-based approaches \citep[e.g.,][]{hdoutliers} instead flag
observations that are far from a centre, using a distance such as
\(|y_i - \hat\mu|/\hat\sigma\) or a Mahalanobis distance, ideally with
robust estimates of location and scale. When the density is unimodal and
elliptically symmetric, the surprisal is a monotone function of the
Mahalanobis distance, so the two approaches flag the same observations.
They differ when the distribution is skewed or multimodal, as low
density then no longer corresponds to large distance from a centre.

The surprisal approach converts the problem of working with \(F\), a
potentially multivariate distribution (or even one on a manifold or some
non-Euclidean space), to working with a univariate distribution \(G\)
defined on \(\mathbb{R}\). We can write \[
  G(s) = \Pr\{-\log f(\bm{Y}) \le s\} = \Pr\{ f(\bm{Y}) \ge e^{-s}\}.
\] This makes clear that the surprisal probability is related to the
highest density regions of the distribution \citep{HDR96}. In fact,
every observation with \(p_i \le \alpha\) lies outside the
\(100(1 - \alpha)\)\% highest density region of \(F\).

If \(F = N(\mu, \sigma^2)\), then
\(p_i = 2[1 - \Phi(|y_i - \mu|/\sigma)]\), where \(\Phi\) is the
standard Normal distribution function. In this case, identifying
anomalies as observations with \(p_i \le \alpha\) is equivalent to a
z-score test with the threshold set to \(\Phi^{-1}(1 - \alpha/2)\).
Hence, our approach can be seen as a generalization of anomaly detection
methods based on z-scores including Grubbs' test \citep{Grubbs1950}.

It can also be seen as a generalization of other anomaly detection
methods. For example, the Hampel identifier \citep[p.~65]{Hampel1986}
flags an observation as an anomaly if \(|y_i - m|/a > \tau\), where
\(m\) is the median and \(a\) is the median absolute deviation (MAD) of
the data. The threshold \(\tau\) is often set under the assumption of a
Normal distribution. When we take a surprisal perspective, this is
equivalent to setting
\begin{equation}\protect\phantomsection\label{eq-hampel}{
  F = N(m, a^2/\Phi^{-2}(0.75)).
}\end{equation} The false anomaly rate is
\(\alpha = 2[1 - \Phi(\tau\,\Phi^{-1}(0.75))]\). The Hampel identifier
is often applied to time series using a moving window
\citep{Davies1993}, so that each observation is compared with its
neighbours. Then \(m\) and \(a\) in (\ref{eq-hampel}) are replaced by
\begin{align*}
  m_t & = \text{median}(y_{t-h}, \dots, y_{t+h}),                 \\
  \text{and}\qquad
  a_t & = \text{median}(|y_{t-h} - m_t|, \dots, |y_{t+h} - m_t|),
\end{align*} giving a different distribution at each time \(t\).

We will usually not know \(F\), and so we need to estimate it from data
or from theory. Such estimates are particularly fragile or uncertain in
areas of low density, which is precisely where we are looking for
anomalies. Further, we are seeking a method that does not make strong
distributional assumptions.

Our approach begins with an assumed distribution \(\widetilde{F}\), with
generalized density \(\tilde{f}\), which may be estimated from data, or
obtained from theory. This distribution is used to compute the surprisal
values \(\tilde{s}_i = -\log \tilde{f}(\bm{y}_i)\), \(i = 1,\dots,n\).

If we also use the assumed distribution \(\widetilde{F}\) to compute the
tail probabilities, we obtain
\(\hat{p}_i = \int \mathds{1}\{\tilde{f}(\bm{y}) \le e^{-\tilde{s}_i}\} d\widetilde{F}(\bm{y})\),
where \(\mathds{1}\) is the indicator function. However, this makes the
strong assumption that \(\widetilde{F}\) is an accurate representation
of the data, and in practice it may be poorly estimated or
mis-specified.

So, instead, we estimate \(p_i\) using one of the following two
approximations:

\begin{enumerate}
\def\labelenumi{\arabic{enumi}.}
\tightlist
\item
  \emph{Empirical approximation}: Use the empirical distribution
  function of the observed surprisals
  \(\{\tilde{s}_1,\dots,\tilde{s}_n\}\), so that \(\hat{p}_i\) is the
  proportion of surprisals at least as large as \(\tilde{s}_i\);
\item
  \emph{GPD approximation}: Fit a Generalized Pareto Distribution (GPD)
  to the largest surprisal values, and use this to estimate the tail
  probabilities \(\hat{p}_i\).
\end{enumerate}

Both approximations allow for considerable mis-specification of the
distribution used to compute the surprisals. The empirical approximation
requires only that the ordering of the surprisals in the tail is similar
to what would have been obtained under the true distribution, while the
GPD approximation also requires that extreme value theory applies to the
surprisal distribution.

In Section~\ref{sec-empirical}, we develop the necessary theory to
support the use of empirical surprisal probabilities, and identify the
conditions under which this will give a good estimate of the true
surprisal probabilities. Similarly, Section~\ref{sec-gpd} provides the
theoretical results needed to use a Generalized Pareto Distribution to
approximate the tail of the surprisal distribution. While we present the
main results in terms of iid data, Section~\ref{sec-conditional} shows
that the methods can be applied to conditional distributions as well. We
demonstrate the methods with several simulation experiments in
Section~\ref{sec-experiments}, and in three applications in
Section~\ref{sec-applications}. A discussion follows in
Section~\ref{sec-discussion}. The appendix contains the proofs needed
for the results in Sections \ref{sec-empirical} and \ref{sec-gpd}.

\section{Empirical surprisal probabilities}\label{sec-empirical}

Under the empirical approximation, we estimate \(p_i\) in (\ref{eq-pi})
by \(\hat{p}_i\), the proportion of observed surprisals
\(\tilde{s}_1,\dots,\tilde{s}_n\) at least as large as \(\tilde{s}_i\).
The anomaly boundary for a threshold of \(\alpha\) corresponds to the
\(100(1 - \alpha)\)\% highest density region computed using the quantile
density approach of \citet{HDR96}.

Let \(F(\cdot)\) denote the true distribution function of \(\bm{Y}\),
and suppose we use \(\widetilde{F}\) to model \(\bm{Y}\). For the
empirical approximation to be accurate, the ordering of large surprisal
values under the fitted model, \(\widetilde{F}\), should match that
under the true distribution, \(F\). Thus, a sufficient condition is that
the true and fitted surprisals are related by a strictly increasing
transformation on the upper tail.

Recall that the surprisal under the true distribution is defined as
\(S = -\log f(\bm{Y})\), and recall that \(G(s) = \Pr(S \le s)\) denotes
the CDF of \(S\) at \(s \in \mathbb{S} \subseteq \mathbb{R}\), where
\(\mathbb{S}\) denotes the support of the surprisals. Likewise, let
\(\widetilde{S} = -\log \tilde{f}(\bm{Y})\) denote the surprisal under
the fitted distribution, and let
\(\widehat{\mathbb{P}}_n(s) := n^{-1}\sum_{i=1}^{n}1[\widetilde{S}_i \le s]\)
denote the empirical CDF of the fitted sample
\(\widetilde{S}_1,\dots,\widetilde{S}_n\).
Assumption~\ref{def-empirical_pis} gives a sufficient condition for the
empirical estimator \(\widehat{\mathbb{P}}_n(s)\) to provide an accurate
estimate of the upper tail of \(G(s)\) in any given sample.

\begin{definition}[]\protect\hypertarget{def-empirical_pis}{}\label{def-empirical_pis}

There exists an \(s_\star \in \mathbb{R}\) and a function \(h(\cdot)\),
strictly increasing on \([s_\star, \infty)\), such that
\(S = h(\widetilde{S})\) almost surely for all
\(\widetilde{S} \ge s_\star\), and \(S \le h(s_\star)\) for
\(\widetilde{S} < s_\star\).

\end{definition}

\begin{lemma}[]\protect\hypertarget{lem-confident}{}\label{lem-confident}

Let \(S_1,\dots, S_n\) be iid with distribution function \(G\). If
Assumption~\ref{def-empirical_pis} is satisfied, then, for any
\(\delta \in (0,1)\), and \(\epsilon = \sqrt{\log(2/\delta)/(2n)}\),
with probability at least \(1 - \delta\), for every
\(\widetilde{S}_i \ge s_\star\) \[
  \widehat{\mathbb{P}}_n(\widetilde{S}_i) - \epsilon \le G(S_i) \le \widehat{\mathbb{P}}_n(\widetilde{S}_i) + \epsilon.
\]

\end{lemma}

Lemma~\ref{lem-confident} implies that when the fitted distribution has
tails that are similarly located to the actual data distribution, then
the fitted empirical surprisal probabilities are accurate measures of
the true surprisal probability, \(G(s)\). The tails of the fitted
distribution need not be the same as those of the data distribution;
beyond a threshold value \(s_\star\), the two need only be related
through a monotone transformation. So the fitted distribution must have
its regions of low density in the same locations as the true data
distribution.

Assumption~\ref{def-empirical_pis} is not satisfied by all pairs of
distributions: if the true data has skewness but the fitted model
assumes symmetry, or if there is bimodality in the observed data that is
not modeled, then Assumption~\ref{def-empirical_pis} may not be
satisfied.

In other cases Assumption~\ref{def-empirical_pis} holds exactly. If the
true and fitted densities are both decreasing functions of the same
distance from the same centre, the two sets of surprisals are monotone
functions of that distance, so \(h\) exists with \(s_\star\) equal to
the smallest surprisal. For example, this is true when a Student's \(t\)
distribution with \(k\) degrees of freedom, denoted \(t(k)\), is fitted
to \(N(0,1)\) data, as in Section~\ref{sec-experiments}, or when a
multivariate normal distribution is fitted to data from an elliptical
distribution with the same location and the same scatter matrix up to
scale. Because \(\hat{p}_i\) depends only on the ranks of the
surprisals, two fitted models that rank the largest surprisals in the
same order give identical empirical probabilities. This gives a simple
practical check: compute the surprisals under two or more plausible
models (for example, a parametric model and a kernel density estimate)
and compare the observations with the largest surprisals. If they
largely agree, the ordering is unlikely to depend on the choice of
model.

Lemma~\ref{lem-confident} also shows how large \(n\) must be to obtain a
given level of accuracy. For fixed \(\delta\), the half-width
\(\epsilon\) of the band shrinks to zero like \(n^{-1/2}\) and does not
depend on the tail probability of the true distribution. Hence,
Lemma~\ref{lem-confident} is informative only when \(\epsilon\) is small
compared with the tail probabilities of interest. With \(n = 1000\), the
95\% bound gives \(\epsilon \approx 0.043\), which is larger than a
threshold of \(0.01\). A pointwise argument is less conservative: at a
true tail probability of \(\alpha\), \(\hat{p}_i\) is a binomial
proportion with relative standard error
\(\sqrt{(1 - \alpha)/(n\alpha)} \approx (n\alpha)^{-1/2}\). So the
accuracy of the empirical approximation is governed by \(n\alpha\), the
expected number of observations beyond the threshold, rather than by
\(n\) alone. As a rough guide, we suggest \(n\alpha \ge 10\) (a relative
standard error of about 30\%) for the empirical approximation. When
\(n\alpha\) is smaller, the GPD approximation of Section~\ref{sec-gpd}
is preferable, because it uses the largest \(100 \times \beta\%\) of the
surprisals rather than only those beyond the threshold.

\section{Extreme value theory for the surprisal
distribution}\label{sec-gpd}

The GPD approximation uses extreme value theory applied to the
(univariate) surprisal values. We will show that under weak conditions
on the tail of the surprisal distribution, the excesses of the surprisal
over a high threshold converge to a GPD, with the surprisal distribution
in either the Gumbel or Fréchet domain of attraction. This justifies
fitting a Generalized Pareto Distribution (GPD) to the largest surprisal
values, and using it to estimate tail probabilities.

An important precursor to this work is \citet{lookout} who proposed
estimating \(F\) by a leave-one-out kernel density estimate, and then
fitting a Generalized Pareto Distribution (GPD) to the most extreme
surprisal values. The underlying idea is that even if \(F\) is not
particularly well-estimated, the tail of the surprisal distribution
should still be well-approximated by a GPD. \citet{lookout2} provide
some theoretical justification for this approach in the context of
kernel density estimation.

Here we extend that idea to allow \(F\) to be any continuous
distribution, including conditional distributions; certain discrete
distributions can be handled at the cost of additional technicalities,
while others would not deliver surprisals with a continuous GPD limit.
Specifically, we set the threshold \(u\) to be the \(100(1 - \beta)\)\%
sample quantile of \(\{\tilde{s}_i\}_{i=1}^n\), and fit a GPD to the
excesses \(\tilde{s}_i - u\) of the surprisals above \(u\), obtaining
estimates \(\hat\omega\) and \(\hat\xi\) of the scale and shape
parameters. We then estimate the probability of a surprisal at least as
large as that of each observation by
\begin{equation}\protect\phantomsection\label{eq-gpdprob}{
  \hat{p}_i = \beta\left[ 1 - H_{\hat\xi}\{(\tilde{s}_i - u)/\hat\omega\} \right],
}\end{equation} where \(H_\xi\) is the GPD distribution function given
in (\ref{eq-gpd}). Observations with \(\tilde{s}_i \le u\) are assigned
\(\hat{p}_i = \beta\).

To justify this approach, specifically to ensure a Pickands--Balkema--de
Haan theorem is satisfied, we seek conditions on the distribution of the
surprisal \(\widetilde{S} := -\log \tilde{f}(\bm{Y})\), where
\(\bm{Y} \sim F\), and \(\tilde{f}\) is the assumed or estimated density
function, which may be distinct from \(f\). We formulate our results
using moments of \(\widetilde{S}\), as well as additional regularity
conditions. We believe this formulation is helpful as the moment
behavior of \(\widetilde{S}\) can often be directly discerned at a
glance and so does not require that the practitioner have previous
knowledge of EVT or any knowledge of the domain of attraction to which
\(\widetilde{S}\) belongs. For instance, if we know that
\(Y \sim N(0,1)\) and \(\tilde{f} = f\), then it is immediately obvious
that \(\widetilde{S}\) is a shifted and scaled chi-squared random
variable, and it is taught in undergraduate statistics classes that
chi-squared random variables have an exponential moment.

To state our results, we thus maintain the following conditions. They
are stated for a generic random variable \(S\), and we will apply them
with \(S = \widetilde{S}\).

\begin{definition}[Exponential Moment
Existence]\protect\hypertarget{def-bounded}{}\label{def-bounded}

The random variable \(S \in \mathbb{R}\) is such that \(\mu = \E[S]\)
exists. For some parameters \(\nu > 0\) and \(\sigma^2 > 0\), and any
\(\lambda \in \mathbb{R}\) such that \(|\lambda| \le 1/\nu\),
\(\E[e^{\lambda(S-\mu)}] \le e^{\lambda^2\sigma^2/2}.\)

\end{definition}

\begin{definition}[Polynomial Moment
Existence]\protect\hypertarget{def-poly}{}\label{def-poly}

For some \(q > 1\), \(\E[|S-\mu|^{q}] < \infty\).

\end{definition}

Assumption~\ref{def-bounded} ensures that \(S\) is a sub-exponential
random variable, and is easy to check for many distributions. For
example, if \(Y \sim t(k)\), \(k \ge 1\), then \(S = -\log f(Y)\) has a
moment generating function, on a certain domain, which implies
Assumption~\ref{def-bounded} is satisfied: the quantities
\(\mu,\sigma^2\) and \(\nu\) can be directly calculated from the MGF and
are functions of the degrees of freedom parameter \(k\).

Assumption~\ref{def-poly} will be satisfied in cases where \(S\) only
admits a polynomial moment. For instance, if we compute the surprisals
using a standard Gaussian density \(\tilde{f}\), so that
\(\widetilde{S} = c + Y^2/2\), but in fact \(Y \sim t(k)\), then
\(\widetilde{S}\) will satisfy Assumption~\ref{def-poly} so long as
\(k > 2\); i.e., provided the observed data have a finite variance.

Define the GPD distribution function
\begin{equation}\protect\phantomsection\label{eq-gpd}{
  H_{\xi}(y)=\left\{\begin{array}{lll}
  1-(1+\xi y)^{-1 / \xi}, & \xi \neq 0, & 1+\xi y>0 \\
  1-e^{-y}, & \xi=0. &
  \end{array}\right.
}\end{equation}

\begin{theorem}[]\protect\hypertarget{thm-gpd}{}\label{thm-gpd}

Let \(\bm{Y}_1,\dots,\bm{Y}_n\) be an iid sequence from \(F\), and let
\(\widetilde{S} = -\log \tilde{f}(\bm{Y})\), where \(\bm{Y} \sim F\),
with \(\mu = \E[\widetilde{S}]\). Let the regularity conditions in
Assumption~\ref{def-regular} be satisfied with \(S = \widetilde{S}\).
Then, there exists a scale function \(\omega(u) > 0\) such that \[
  \lim_{u \uparrow \infty}\sup_{y \ge 0} \left| \Pr\left( \frac{(\widetilde{S} - \mu) - u}{\omega(u)} \le y\,\Big|\,(\widetilde{S} - \mu) > u \right) - H_\xi(y) \right| = 0,
\] with \(\xi = 0\) under Assumption~\ref{def-bounded} or
\(0 \le \xi \le 1/q\) under Assumption~\ref{def-poly}.

\end{theorem}

Theorem~\ref{thm-gpd} shows that the conditional excess distribution
function of the surprisal converges to the GPD where the value of
\(\xi\) is constrained by the number of finite moments that
\(\widetilde{S}\) possesses. If the surprisal distribution has at least
an exponential moment and full support, then the conditional excess
distribution function of the surprisal behaves like a shifted
exponential distribution; if instead the surprisal only has a polynomial
moment, when \(\xi \ne 0\) the conditional excess distribution function
converges to the Pareto distribution. Theorem~\ref{thm-gpd} therefore
justifies the peaks-over-threshold approach: the fitted GPD should give
an accurate approximation to the true surprisal probability in large
samples, provided we can accurately estimate the parameters of the
Generalized Pareto Distribution \citep[Theorem
4.1]{coles2001introduction}.

It is possible to link Assumptions \ref{def-bounded} and \ref{def-poly}
to the domain of attraction conditions usually considered in extreme
value theory. For simplicity, let us assume that range of
\(\widetilde{S}\) is \(\mathbb{R}\). Under technical conditions,
Assumption~\ref{def-bounded} then implies that the distribution of
\(\widetilde{S}\), denoted by \(\widetilde{G}\) herein, is in the Gumbel
domain of attraction, i.e., \(H_0\) in Equation~\ref{eq-gpd}. In
contrast, under Assumption~\ref{def-poly}, \(\widetilde{G}\) can be in
either the Gumbel or Fréchet domain of attraction: where \(q\) denotes
the number of assumed moments, we only have that \(0 \le \xi \le 1/q\)
and so cannot exclude the Gumbel domain of attraction; if, however, we
know that the \((q + 1)\)-th moment of \((\widetilde{S} - \mu)_{+}\) is
infinite, then under Assumption~\ref{def-poly} we have that
\(\widetilde{G}\) is in the Fréchet domain of attraction and
\(0 < 1/(1 + q) \le \xi \le 1/q\).

Assumptions \ref{def-bounded} and \ref{def-poly} are convenient
sufficient conditions, but they are not necessary for the GPD
approximation. The GPD limit for threshold excesses holds whenever
\(\widetilde{G}\) is in the domain of attraction of an extreme value
distribution \citep[Theorem 4.1]{coles2001introduction}, and this
requires no moments. Because Assumption~\ref{def-poly} needs \(q > 1\),
Theorem~\ref{thm-gpd} only allows \(\xi < 1\), and so excludes surprisal
distributions without a finite mean. These can still be handled by the
GPD approximation. For example, if the surprisals are computed from a
Gaussian density but \(Y \sim t(1)\), then \(\widetilde{S} = c + Y^2/2\)
has no finite mean, but
\(\Pr(\widetilde{S} - c > s) = \Pr(|Y| > \sqrt{2s})\) is proportional to
\(s^{-1/2}\) for large \(s\), so \(\widetilde{G}\) is in the Fréchet
domain of attraction with \(\xi = 2\). The cost of such heavy tails is
statistical rather than theoretical: \(\xi\) is harder to estimate, and
extrapolation beyond the data is less reliable.

\section{Applicability for conditional
distributions}\label{sec-conditional}

The theoretical results in Lemma~\ref{lem-confident} and
Theorem~\ref{thm-gpd} assume that the surprisal is based on iid data
\(\bm{Y}\) generated from some unknown distribution. In this way, these
results are agnostic to the dimension of \(\bm{Y}\); the resulting
surprisal \(\widetilde{S}_i = -\log \tilde{f}(\bm{Y}_i)\) is always
univariate, and remains iid since any function of \(\bm{Y}\) is also
iid.

However, there are many modeling situations, such as regression
settings, where we are not necessarily interested in measuring the
surprisal that results from the joint distribution of the data, but
along a fixed vector of conditioning variables (e.g., age, or gender).
In such cases, we can think of our observed data being the variables
\(\bm{Y} = (Z,\bm{X})\) where \(Z\) is some outcome of interest and
\(\bm{X}\) are variables that are used to predict or explain \(Z\). In
such cases, we can directly apply the surprisal approach with a
conditional model of the form \(\tilde{f}_Z(z \mid \bm{X} = \bm{x})\),
and the surprisal
\(\widetilde{S} = -\log \tilde{f}_Z(Z \mid \bm{X} = \bm{x})\) can be
thought of as being computed ``conditional on a given
\(\bm{X} = \bm{x}\)''.

For a fixed vector of the conditioning variables, it is not hard to see
why our theoretical results will remain valid: by fixing
\(\bm{X} = \bm{x}\), since we observe iid \(\bm{Y} = (Z_i,\bm{X}_i)\),
the distribution of \(Z \mid \bm{X} = \bm{x}\) is also iid conditional
on \(\bm{X} = \bm{x}\), and the resulting theory applies. That is, for
example, if we wish to calculate surprisals at a fixed value of the
covariates \(\bm{x}\), then Lemma~\ref{lem-confident} and
Theorem~\ref{thm-gpd} can be applied conditional on this fixed set of
covariates.

However, if we wish to calculate surprisals across both the outcomes and
covariates (i.e., when allowing both \(Z\) and \(\bm{X}\) to vary), then
what is a surprisal may change from instance to instance. To see why,
consider the simple example where we wish to detect surprisals using a
simple Gaussian regression model for \(Z \mid X\), so that \[
  \widetilde{S}_i = -\log \tilde{f}_Z(Z_i \mid X_i) = c + (Z_i - \beta_0 - \beta_1 X_i)^2/(2\sigma_Z^2).
\] Now, \(\widetilde{S}_i\) can be large for two distinct and
interrelated reasons: either because \(Z\) itself is large, relative to
the observed values of \(\beta_0 + \beta_1 X\), or because \(X\) is
atypical, where the assumed regression may be a poor approximation, so
that \(Z - \beta_0 - \beta_1 X\) is large.

That is, in the general case, where we are searching for surprisals over
the conditional model without \(\bm{X}\) being fixed, the behavior of
\(\widetilde{S}_i\) is determined by that of the joint distribution: \[
  \Pr(\widetilde{S}_i > s) = \int \mathds{1}(-\log \tilde{f}_Z(z \mid \bm{x}) > s) f_Z(z \mid \bm{x}) f_{\bm{X}}(\bm{x})\,dz\,d\bm{x},
\] where \(f_Z\) and \(f_{\bm{X}}\) denote the true conditional and
marginal densities. This clarifies that extreme behavior in
\(\widetilde{S}_i\) can result from extreme values of \(Z\), \(\bm{X}\),
or simply values of \(Z\) and \(\bm{X}\) where
\(\tilde{f}_Z(Z \mid \bm{X})\) is very small. The latter can occur even
when \(Z\) and \(\bm{X}\) take on values that are typical for their
marginal distribution, but which are not well-captured by the assumed
conditional model.

This determines which tail probability is estimated. When surprisals are
pooled over observations with different covariate values, as in the
French mortality and cricket applications, the empirical and GPD
approximations estimate the marginal tail probability
\(\Pr(\widetilde{S} \ge \tilde{s}_i)\) given above, where
\((Z, \bm{X})\) varies over its joint distribution. This is an average
over \(\bm{x}\) of the conditional tail probabilities
\(\Pr(\widetilde{S} \ge \tilde{s}_i \mid \bm{X} = \bm{x})\). The two
coincide when the conditional distribution of \(\widetilde{S}\) given
\(\bm{X} = \bm{x}\) does not depend on \(\bm{x}\). This holds when the
surprisal depends on the data only through a standardized residual whose
distribution does not depend on \(\bm{x}\). Examples are the Gaussian
regression above when the model is correct, where
\(2(\widetilde{S}_i - c)\) has a \(\chi^2_1\) distribution at every
\(\bm{x}\), and the French mortality application in
Section~\ref{sec-applications}, where the surprisal is computed from the
standardized residual of each series. When the conditional distribution
of \(\widetilde{S}\) depends on \(\bm{x}\), as in the Binomial model for
the cricket data where the number of innings differs between batters, an
observation is flagged when its surprisal is unusual relative to all
observations, not only those with the same covariate values. Estimating
the conditional tail at a fixed \(\bm{x}\) requires many observations at
(or near) that value of \(\bm{x}\), or else computing the tail
probability directly from the conditional model, which relies on that
model being correct.

\section{Experiments}\label{sec-experiments}

\subsection{Standard univariate normal
values}\label{standard-univariate-normal-values}

Although we need to assume some density function \(f\) when computing
surprisals, we want the resulting tail probabilities to be robust to a
mis-specified distribution. Suppose we generate 1000 observations from a
\(N(0,1)\) distribution \(f\), but use a \(t(4)\) distribution
\(\widetilde{f}\) to compute the surprisal values. The empirical tail
probabilities, and those obtained using the GPD approximation, should
still give reasonable estimates of the tail probabilities, despite the
incorrect distribution being used to compute the surprisals.

To study the effect of the assumed distribution, we generate 1000
observations from each of the \(N(0,1)\), \(t(4)\) and \(t(1)\)
distributions, and compute the surprisal values using both a standard
Normal density and a \(t(4)\) density. Then, for each set of surprisal
values, we estimate the surprisal tail probabilities in three ways: (a)
using the same distribution used to compute the surprisal values; (b)
using the empirical approximation; and (c) using the GPD approximation
fitted to the largest 100 surprisal values.

\begin{figure}[!htb]

\centering{

\pandocbounded{\includegraphics[keepaspectratio]{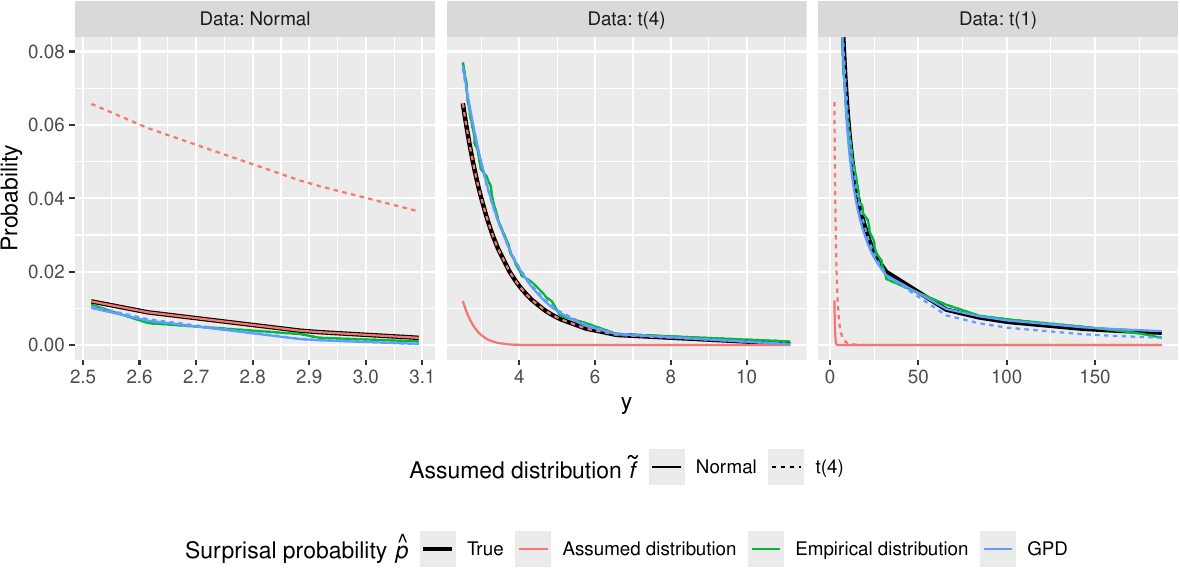}}

}

\caption{\label{fig-expt1}Left: Observations are generated from a
\(N(0,1)\) distribution. Middle: Observations are generated from a
\(t(4)\) distribution. Right: Observations are generated from a \(t(1)\)
distribution. Assumed distribution \(\widetilde{f}\) indicates which
distribution was used to compute the surprisal values. Surprisal
probability \(\hat{p}\) shows how the tail probabilities were computed.
This demonstrates that the surprisal probabilities computed using the
empirical distribution or under the GPD are still relatively accurate,
even when the wrong distribution is used to compute the surprisal
values.}

\end{figure}%

Figure~\ref{fig-expt1} shows the estimates for observations greater than
2.5. In the left panel (\(N(0,1)\) data), the probabilities computed
directly from the incorrect \(t(4)\) density are far too large, but the
empirical and GPD approximations are close to the truth under both
assumed densities. The middle panel (\(t(4)\) data) shows the converse:
the probabilities computed directly from the incorrect Normal density
are far too small, but the empirical and GPD approximations remain
relatively accurate. In the right panel (\(t(1)\) data), neither assumed
density is correct. Gaussian surprisals do not have a finite mean in
this case, so Theorem~\ref{thm-gpd} does not apply, but, as noted in
Section~\ref{sec-gpd}, they are in the Fréchet domain of attraction with
\(\xi = 2\), so the GPD approximation is still justified. The \(t(4)\)
surprisals satisfy Assumption~\ref{def-bounded}. Nevertheless, the
empirical and GPD approximations are again close to the truth under both
assumed densities, while the probabilities computed directly from either
density are far too small.

So even when we don't know the distribution of the data, we can use an
incorrect assumed distribution to compute the surprisals, and then apply
either the GPD or the empirical approximation to obtain good estimates
of the surprisal probabilities.

\subsection{Bivariate gamma values}\label{bivariate-gamma-values}

The quality of the approximation used depends on the amount of data
available, so in this next experiment we compare the tail probabilities
of the different approaches across sample sizes. Let
\(\bm{Y} = (Y_{i,1},Y_{i,2})'\) where \(Y_{i,j}\) are independent
Gamma(2,2) random variables. We generate \(n\) observations from this
bivariate distribution, and explore the quality of the probability
estimates as \(n\) increases from \(100\) to \(10,000\). The surprisal
values are computed using the correct distribution, as well as three
incorrect distributions, each with independent components: a bivariate
normal distribution with mean \((1,1)'\) and covariance matrix
\(0.5\bm{I}\); noncentral \(t(4)\) components with mean 1; and
Exponential(1) components (i.e., Gamma(1,1)). For each set of surprisal
values, we estimate the tail probabilities in two ways: (a) using the
empirical approximation; and (b) using the GPD approximation fitted to
the largest 10\% of surprisal values. We repeat this experiment 1000
times for each sample size. The anomaly rate shown is a loess smooth of
the proportion of anomalies detected when \(\alpha = 0.01\) for
different sample sizes, along with 95\% confidence intervals. The false
anomaly rate under the true distribution is, of course, 0.01. The
supplementary material repeats the experiment with the GPD fitted to the
largest 5\% and 20\% of surprisal values; at \(\alpha = 0.01\) the
results are almost unchanged.

\begin{figure}

\centering{

\pandocbounded{\includegraphics[keepaspectratio]{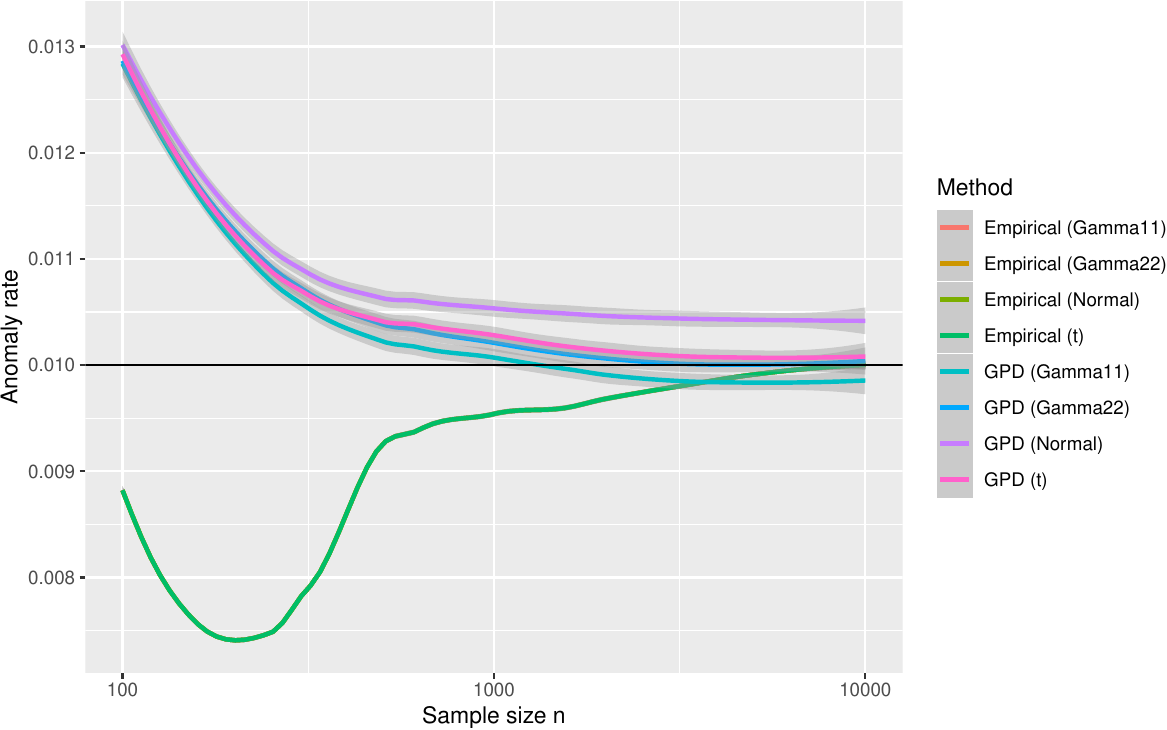}}

}

\caption{\label{fig-expt2}Estimated false anomaly rate under different
approximations when \(\alpha = 0.01\).}

\end{figure}%

The empirical approximation always flags exactly
\(\lfloor n\alpha \rfloor\) observations, whatever distribution is used
to compute the surprisals, so its false anomaly rate is
\(\lfloor n\alpha \rfloor / n\) under all four assumed distributions.
This falls below \(\alpha\) whenever \(n\alpha\) is not an integer,
which, rather than estimation error, explains its shortfall at small
\(n\). The false anomaly rate does not show whether the same
observations are flagged. The Gamma and Normal surprisals are not
monotonically related, even in the tail, because the Gamma surprisal is
unbounded as either component approaches zero while the Normal surprisal
is not. The GPD approximation flags slightly too many anomalies in small
samples, but this bias disappears as \(n\) increases. Among the GPD
fits, using a Student-t reference distribution performs nearly as well
as using the true model, whereas using a Normal reference distribution
is noticeably worse even at large sample sizes. The reason is tail
mismatch: the Normal distribution is lighter-tailed than the true Gamma
distribution, while the Student-t is heavier-tailed. This suggests that
it is safer to err on the side of assuming heavier tails, as
underestimating tail heaviness leads to slow convergence, while
overestimating tail heaviness incurs relatively little cost.

\subsection{Inlying anomalies in bimodal
data}\label{inlying-anomalies-in-bimodal-data}

The previous experiments only involve anomalies in the tails of the
distribution. The surprisal approach also finds inlying anomalies, such
as observations in the low-density gap between the modes of a multimodal
distribution. To demonstrate this, we generate \(n = 1000\)
observations, of which 990 are drawn from the bimodal mixture
\(\tfrac12 N(-4,1) + \tfrac12 N(4,1)\), and the remaining 10 are inlying
anomalies drawn from \(N(0,0.5^2)\), which lie in the gap between the
two modes.

We compare six methods. In each case the anomaly probability is
estimated and an observation is flagged as an anomaly if the probability
is at most \(\alpha = 0.01\). The first four use surprisals, computed
from either the true mixture density or a two-component Gaussian mixture
fitted to the data by maximum likelihood (the fit includes the
anomalies, so it is itself slightly mis-specified), with tail
probabilities estimated using either the empirical approximation or the
GPD approximation fitted to the largest 10\% of surprisal values. The
fifth is the local outlier factor (LOF) of \citet{Breunig2000}, a
standard density-based method that compares the local density around
each observation to that of its neighbours. We use \(k = 50\) nearest
neighbours (5\% of \(n\)), and flag the 1\% of observations with the
largest LOF scores, so that it can be compared with the other methods at
the same level. The sixth is the Hampel identifier, a robust
distance-based method, with \(p_i\) computed from the Normal
distribution in (\ref{eq-hampel}) using the median and MAD of the data.

\begin{figure}[!htb]

\centering{

\pandocbounded{\includegraphics[keepaspectratio]{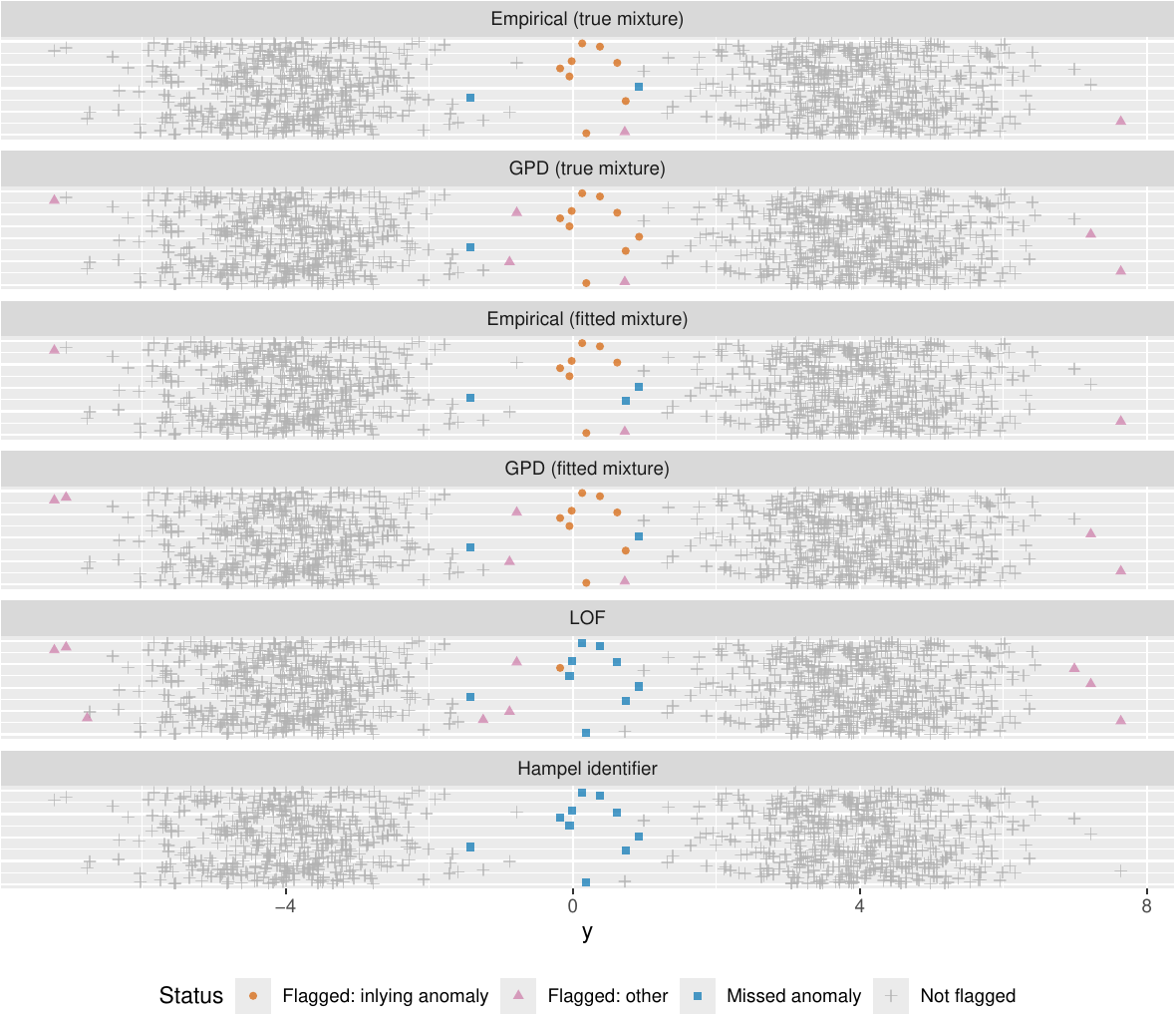}}

}

\caption{\label{fig-expt3}One simulated data set from the bimodal
experiment, with 10 inlying anomalies between the modes. Each panel
shows the observations flagged as anomalies at \(\alpha=0.01\) by one
method. Points are jittered vertically to reduce overplotting.}

\end{figure}%

\begin{table}[!htb]

\caption{\label{tbl-expt3}Average detection rate of the 10 inlying
anomalies, and average false anomaly rate among the 990 observations
from the mixture, over 500 simulated data sets with \(\alpha=0.01\).}

\centering{

\begin{tabular}{lrr}
\toprule
Method & Detection rate & False anomaly rate\\
\midrule
Empirical (fitted mixture) & 0.874 & 0.001\\
Empirical (true mixture) & 0.881 & 0.001\\
GPD (fitted mixture) & 0.959 & 0.003\\
GPD (true mixture) & 0.964 & 0.003\\
Hampel identifier & 0.000 & 0.000\\
LOF & 0.514 & 0.005\\
\bottomrule
\end{tabular}

}

\end{table}%

Figure~\ref{fig-expt3} shows one simulated data set, and
Table~\ref{tbl-expt3} summarizes 500 repetitions. LOF finds only about
half of the inlying anomalies. It is also sensitive to \(k\): with small
\(k\) the 10 anomalies are each other's nearest neighbours, so they look
like a small cluster rather than anomalies, and in a preliminary run
with \(k = 5\) and \(k = 10\) almost none were found. The Hampel
identifier finds none of the inlying anomalies, and flags no
observations at all. The inlying anomalies lie near the median, where
the distance is smallest, and the MAD of the bimodal data (about 6) is
so large that even observations in the outer tails of the modes are not
far enough from the median to be flagged. In contrast, all four
surprisal-based methods identify most of the inlying anomalies, because
the density is very low in the gap between the modes, and they do not
depend on a neighbourhood size. The few anomalies that are missed are
those furthest from zero, which lie closer to the modes and so are not
truly unusual under the mixture. The empirical approximation is limited
by \(n\alpha = 10\): at most about 10 observations can be flagged, and
some of these are taken up by extreme observations in the tails of the
modes. The GPD approximation detects more of the inlying anomalies at
the cost of a slightly higher false anomaly rate, which remains well
below \(\alpha\). The results using the fitted mixture are almost
identical to those using the true mixture, so the density need not be
known.

\section{Applications}\label{sec-applications}

\subsection{French mortality rates}\label{french-mortality-rates}

French mortality rates, disaggregated by age (single years) and sex, are
available from the \citet{HMD}. We will consider male and female data
from 1816 to 1999 over ages 0 to 85, giving 172 separate time series,
each of length 184.

Figure~\ref{fig-fr_mortality_plot} shows the data for both sexes and all
ages. The mortality rates have improved over time for all ages,
especially since 1950. Infant mortality (age 0, shown as the top purple
line) is much higher than the rest of childhood, with mortality rates at
a minimum at about age 10 in all years. The highest mortality rates are
for the oldest age groups. The effects of the two wars are particularly
evident in the male mortality rates.

\begin{figure}

\centering{

\pandocbounded{\includegraphics[keepaspectratio]{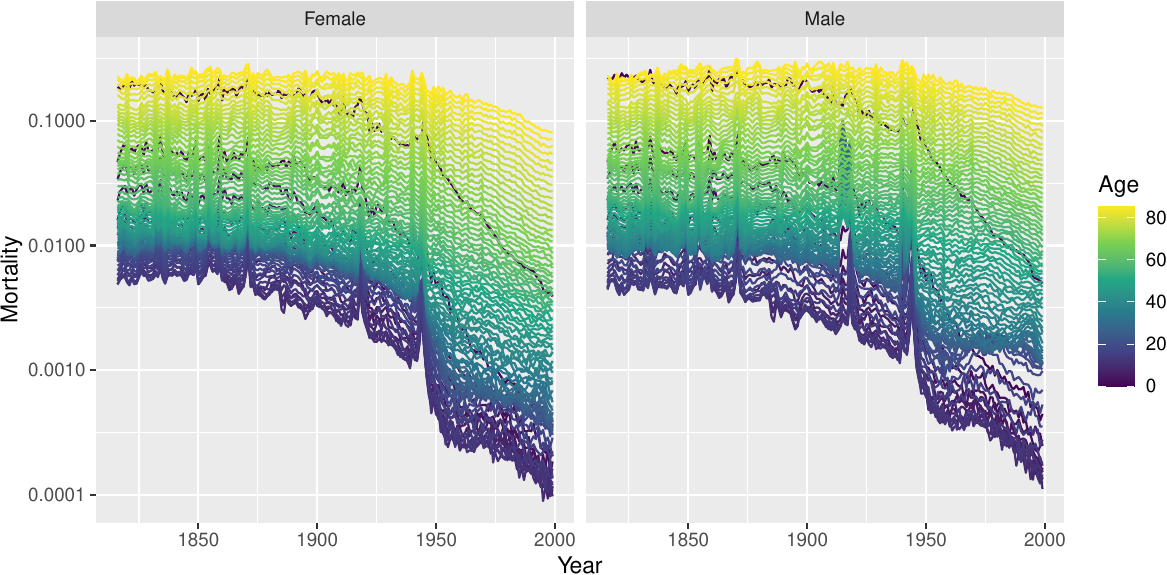}}

}

\caption{\label{fig-fr_mortality_plot}French mortality rates by sex and
age from 1816 to 1999. We use a log scale because the rates are vastly
different for different age groups.}

\end{figure}%

\begin{figure}[!htb]

\centering{

\pandocbounded{\includegraphics[keepaspectratio]{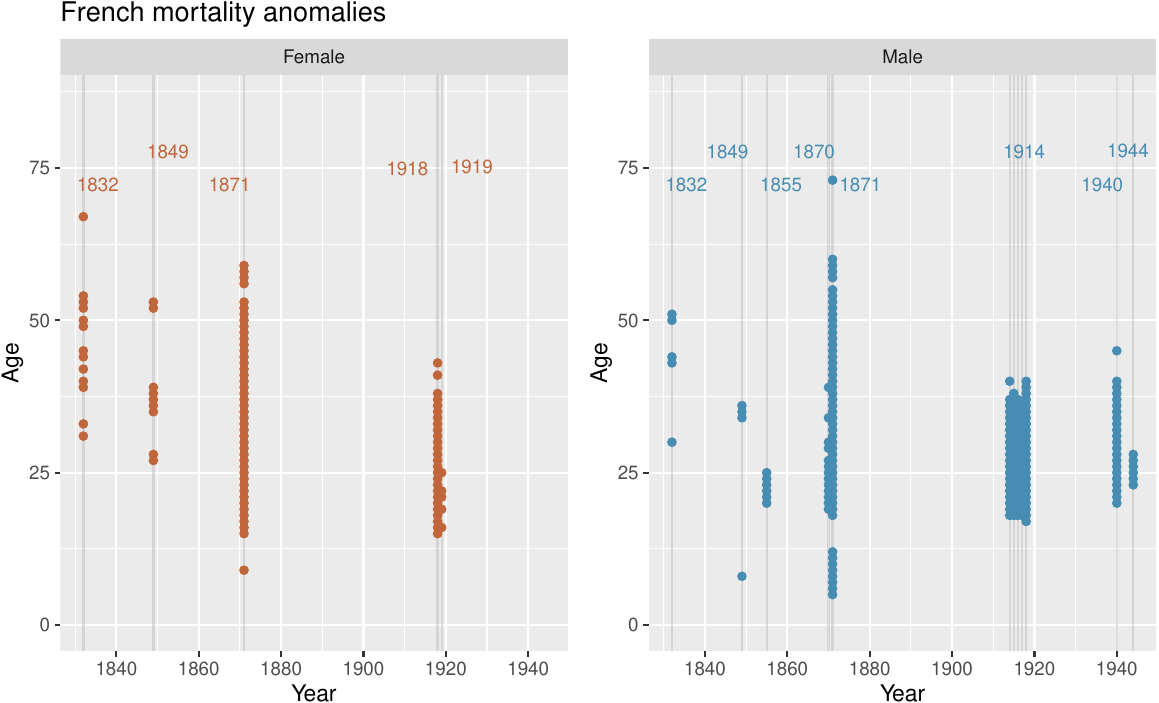}}

}

\caption{\label{fig-hampel_fr}Anomalies identified in the French
mortality data by year and age. They coincide with wars and epidemics in
French history.}

\end{figure}%

We treat each age--sex series as a separate univariate time series, and
use the moving-window version of the Hampel model (\ref{eq-hampel}) with
\(h = 7\), so the running median \(m_t\) and MAD \(a_t\) are computed
over a 15-year window. The first and last 7 years of each series are not
tested, as the window is incomplete. Because each MAD is computed from
only 15 observations, the standardised residuals \((y_t - m_t)/a_t\)
have much heavier tails than a Normal distribution. Following the advice
in Section~\ref{sec-experiments} to err on the side of heavier tails, we
replace the Normal distribution in (\ref{eq-hampel}) by a \(t(4)\)
distribution with the same median and MAD, and compute the surprisal of
each standardised residual under this distribution.

Under the model, the standardised residuals have the same distribution
in every series, and so do their surprisals. So rather than estimating
the surprisal tail probabilities separately within each series, from
only 170 values, we pool the surprisals from all series, giving
\(n = 29{,}240\). We then compute the surprisal probabilities using the
GPD approximation fitted to the largest 10\% of the surprisals, so
\(n\beta \approx 2{,}924\), with \(\alpha = 0.01\). With so many
observations, there is a high likelihood of many false positives (about
292 are expected), so we only retain anomalies when at least three age
groups are anomalous in the same year, and for the same sex. This
results in 309 anomalies out of 31648 observations.
Figure~\ref{fig-hampel_fr} shows these anomalies, which coincide with
various wars and epidemics that have occurred over French history
\citep{Bourdelais1987, french_mortality}:

\begin{itemize}
\tightlist
\item
  1832, 1849: Cholera outbreaks
\item
  1855: Crimean War (males aged 20--25)
\item
  1870: Franco--Prussian war
\item
  1871: Repression of the `Commune de Paris'
\item
  1914--1918: World War I
\item
  1918--1919: Spanish flu outbreak
\item
  1940, 1944: World War II
\end{itemize}

The empirical approximation gives almost the same result. With
\(n\alpha \approx 292\), it is well within the guideline of
Section~\ref{sec-empirical}. Using it instead of the GPD approximation,
the same filtering rule gives 287 anomalies rather than 309, all of
which are also found by the GPD approximation, and exactly the same
years are identified. Supplementary material gives extreme value
diagnostics for the GPD fit, and shows that the anomalous years are the
same for thresholds between \(\beta = 0.03\) and \(\beta = 0.1\).

The theory in Sections \ref{sec-empirical} and \ref{sec-gpd} assumes
independent observations, but these surprisals are dependent. Within
each series, the median lag-one autocorrelation of the surprisals is
0.21, and a surprisal above the 90\% quantile is followed by another in
the next year with probability 0.25, rather than 0.1. The surprisals are
also correlated across ages in the same year, most strongly in years of
war and epidemic. This dependence does not change what is being
estimated: \(p_i\) is a marginal tail probability, and for an
approximately stationary sequence with weak serial dependence, the
empirical distribution of the surprisals still estimates the marginal
distribution consistently, and the threshold excesses still have an
approximate GPD marginal distribution \citep[Ch.
5]{coles2001introduction}. It does reduce the precision. The effective
sample size is smaller than \(n\), so the bound in
Lemma~\ref{lem-confident} is optimistic, and the GPD fit treats the
exceedances as independent and so understates the uncertainty of its
parameters. Declustering or estimating the extremal index would be
needed to model the clusters themselves, for example to estimate the
probability of an anomalous year, but not to estimate \(p_i\). The
filtering rule above makes use of the dependence across ages: a war or
epidemic produces anomalies in many age groups in the same year, whereas
false positives are scattered.

\subsection{Test cricket not outs}\label{test-cricket-not-outs}

Career batting data for all men to have played Test cricket between 1877
and 6 October 2025 were obtained from \citet{Rweird}. We will
investigate if there are some batters who have an unusually high
proportion of innings where they have not been dismissed, known as a
``not out''. The data set contains results from 92646 innings, of which
11943 were not outs. So the overall proportion of not outs is
\(11943 / 92646 = 0.129\).

Batters who bat early in the batting order are more likely to be
dismissed than those who bat last, but the data do not record batting
position. However, the best batters tend to bat early, leaving the least
skilled batters (usually specialist bowlers) to bat last. A measure of
batting skill is the number of runs scored per innings (RPI), the
uncensored batting average. (Unlike the conventional batting average,
RPI is not directly affected by the number of not outs.)
Figure~\ref{fig-notouts} shows the proportion of not outs for each
batter against their RPI. The relationship is slightly U-shaped: weaker
batters are disproportionately likely to be the ``last man standing'',
while the very best batters are harder to dismiss. The horizontal
patterns in the grey points are due to the discrete nature of the data,
as batters with few innings can only have a small number of possible
proportions.

\begin{figure}[!htb]

\centering{

\pandocbounded{\includegraphics[keepaspectratio]{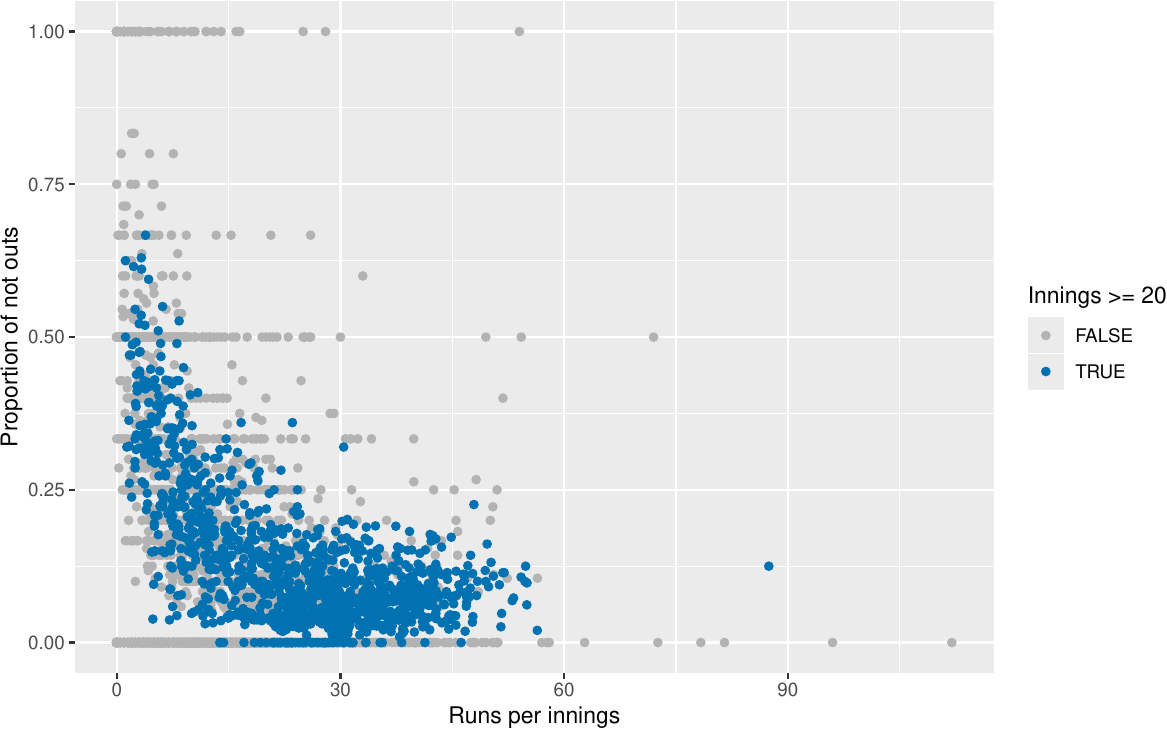}}

}

\caption{\label{fig-notouts}Proportion of not outs for each batter as a
function of their runs per innings. Batters with fewer than 20 innings
are shown in grey.}

\end{figure}%

The style of play and playing conditions have also changed considerably
over almost 150 years of Test cricket. To account for these changes, and
for the variation in batter skills, we construct a generalized additive
model \citep{Wood2017gam} for the number of not outs as a function of
RPI and the year in which the batter started playing. It is natural to
use a Binomial distribution with a logit link function: \[
  \text{NotOuts} \mid \text{RPI}, \text{Start} \sim \text{Binomial}(n = \text{Innings},~ q),
\] where \(q\) denotes the probability of a batter not being dismissed
in an innings, \[
  \log(q / (1 - q)) = g_1(\text{RPI}) + g_2(\text{Start}),
\] and \(g_1()\) and \(g_2()\) are smooth functions estimated
nonparametrically using penalized regression splines in the mgcv package
\citep{mgcv}.

\begin{figure}[!htb]

\centering{

\pandocbounded{\includegraphics[keepaspectratio]{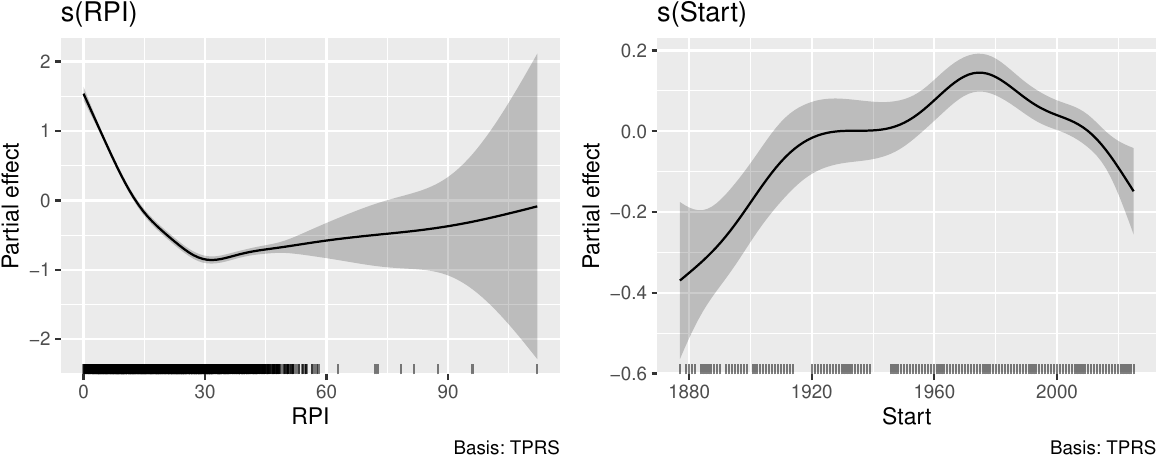}}

}

\caption{\label{fig-notouts-gam}Estimated smooth functions \(g_1\) and
\(g_2\) in the model for not outs, with 95\% confidence intervals.}

\end{figure}%

Figure~\ref{fig-notouts-gam} shows the estimated functions. The
probability of a not out decreases as RPI increases to about 30, and
then increases again for the very best batters, although the confidence
interval is very wide above about 50 as there are few batters with RPI
that high. The probability also increases with the year of debut until
about 1975, and then decreases. This is likely due to the introduction
of limited-overs cricket in the 1970s, after which batters in Test
cricket became more aggressive.

Now we can use the fitted model to compute the surprisals from the
Binomial distribution, and find the most anomalous batters. While the
Binomial distribution is probably a good first approximation, it does
assume that the probability of a batter being dismissed does not change
over their career, which is an unlikely assumption as it does not allow
for the development of skill, the value of experience, or the effect of
aging. So we will use the GPD approximation to compute the surprisal
probabilities. With \(\alpha = 0.01\), 33 of the 3214 batters are
flagged as anomalies; we discuss the two most anomalous.

The largest surprisal is for South African Shaun Pollock, who had 39 not
outs in 156 innings, which is much higher than the expected number of
not outs of \(156 \times 0.091 = 14.1\). Pollock was a bowling
all-rounder and a capable lower-order batter, who often built
partnerships with the weakest batters.

The next most anomalous batter is New Zealander Trent Boult, a
specialist bowler who batted last and showed an unusual flair for
entertaining batting. He had 46 not outs in 94 innings, compared to the
expected number of \(94 \times 0.234 = 22.0\).

Neither batter's proportion of not outs is exceptional by itself (16
batters with at least 20 innings have a higher proportion than Boult),
but they become anomalous once their RPI and year of debut are taken
into account. Because we have used a statistical model, we have been
able to account for the particular features of this data set, such as
the discrete nature of the data, and the changing variance, to identify
observations that are anomalous in the context of the model.

Here the choice of approximation makes almost no difference. With 3214
batters, the empirical approximation gives the same ten most anomalous
batters in the same order, and the probabilities for Pollock and Boult
are 0.0003 and 0.0006, compared with 0.0003 and 0.0004 from the GPD
approximation. The sample is large enough that \(n\alpha\) is not small,
so the empirical approximation is already accurate.

\subsection{Old Faithful eruptions}\label{old-faithful-eruptions}

Our third application is to the Old Faithful geyser in Yellowstone
National Park. The \texttt{oldfaithful} data set in the \textbf{weird}
package \citep{Rweird} contains the duration of each of 2097 eruptions
between 2017 and 2023, along with the waiting time until the following
eruption, both in seconds, compiled from observations recorded at
GeyserTimes (\url{https://geysertimes.org}). The joint distribution of
duration and waiting time is bimodal: there is a cluster of short
eruptions followed by short waits, and a larger cluster of long
eruptions followed by long waits, with a sparsely populated region
between them. Observations in that region are inlying anomalies, and
none of them would be found by a method that only looks for extreme
values of each variable, or one that assumes a unimodal elliptical
distribution.

We model the density with a Gaussian mixture model fitted using the
\textbf{mclust} package \citep{mclust}, which selects the number of
components by BIC, and use it to compute the surprisal of each eruption.
We then estimate the surprisal probabilities using both the empirical
and the GPD approximations (the latter fitted to the largest 10\% of
surprisals), and flag eruptions with probability at most
\(\alpha = 0.01\). The fitted model is itself only an approximation to
the data-generating process, and is estimated from data that include the
anomalies, which is the situation the tail approximations are meant to
handle.

\begin{figure}[!htb]

\centering{

\pandocbounded{\includegraphics[keepaspectratio]{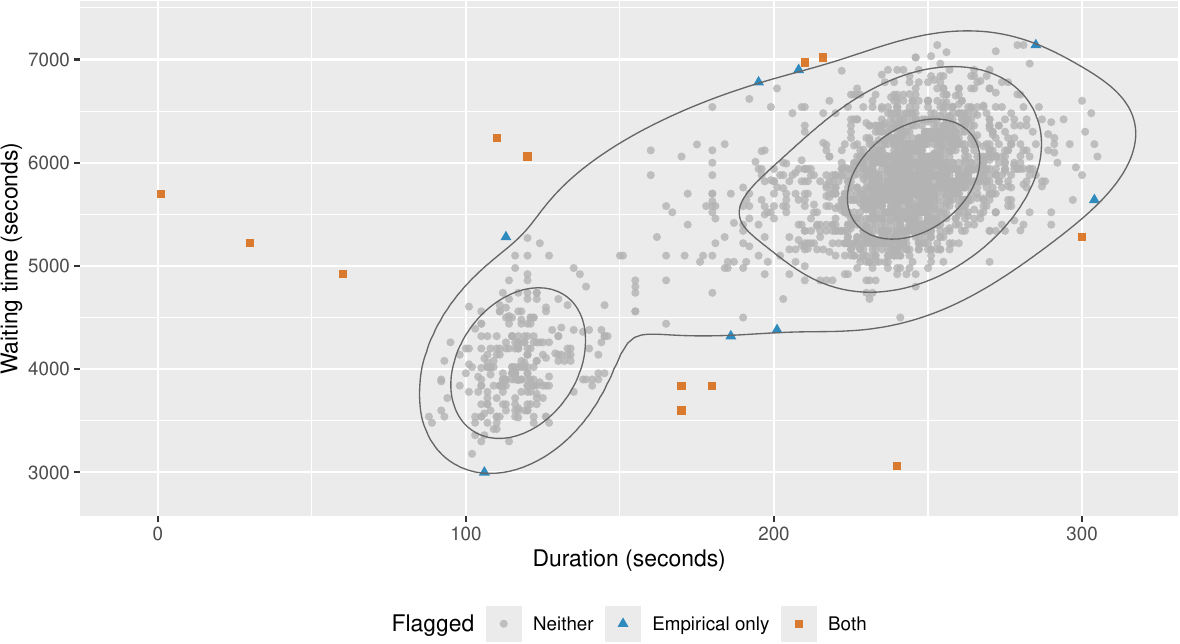}}

}

\caption{\label{fig-oldfaithful}Old Faithful eruption durations and
waiting times. Grey contours show the 50\%, 90\% and 99\% highest
density regions of the fitted Gaussian mixture model. Eruptions with
surprisal probability at most 0.01 are highlighted, according to whether
they were flagged using the GPD approximation, the empirical
approximation, or both.}

\end{figure}%

Figure~\ref{fig-oldfaithful} shows the results. The GPD approximation
flags 12 eruptions and the empirical approximation flags 20; all 12 of
the former are also flagged by the latter. Some of the anomalies are
extreme in at least one variable, such as the eruption apparently
lasting one second, but several lie in the sparse region between the two
clusters, and away from both of them. For example, the three eruptions
of 170 to 180 seconds that were followed by waits of about an hour are
not extreme in either variable on their own, but their combination falls
in a region of very low density. The empirical approximation flags more
observations, probably because its probabilities near the threshold are
noisy, with only about 21 observations below \(\alpha = 0.01\), whereas
the GPD approximation smooths the tail of the surprisal distribution.
The two approximations agree on the most anomalous eruptions.

\section{Discussion}\label{sec-discussion}

The surprisal of an observation measures how unusual it is under a
specified distributional model, and our tail approximations remain
accurate even when that model is mis-specified. Methods that compute
anomaly scores directly from the assumed distribution do not have this
property, as Figure~\ref{fig-expt1} shows. Converting a multivariate (or
otherwise complex) problem into estimating the upper tail of a
univariate surprisal distribution lets us use standard empirical and
extreme-value methods.

Our theoretical results make clear what kind of mis-specification can be
tolerated. Anomalies are defined relative to regions of low density of
\(F\), but are identified using an assumed distribution
\(\widetilde{F}\). For the empirical approximation to be reliable, the
assumed model must correctly identify the \emph{locations} of these
low-density regions, even if it does not accurately capture the
\emph{shape} of the tail. Assumption~\ref{def-empirical_pis} formalizes
this requirement: the true and fitted surprisals need only agree up to a
strictly increasing transformation on the extreme tail. When this
condition holds, Lemma~\ref{lem-confident} guarantees that the empirical
surprisal probabilities form a valid uniform confidence band for the
true tail probabilities on the upper tail. This does not require
\(\widetilde{F}\) to have the correct tail behavior or even the correct
distributional form, only the correct tail \emph{ordering}. When the
condition fails, such as when the assumed model imposes symmetry where
the data are skewed, or unimodality where the data are multimodal, the
empirical approximation no longer provides uniform guarantees. In
practice, such gross mis-specifications should be identifiable through
exploratory data analysis.

The GPD approximation offers a complementary route to estimating
surprisal tail probabilities. Our extreme-value results show that under
broad moment conditions on the surprisal (an exponential moment, or only
a polynomial moment), together with mild regularity conditions, the
excesses of the surprisal over a high threshold converge to a
Generalized Pareto Distribution. These conditions encompass almost all
common statistical models. Thus, the GPD approximation can be expected
to estimate the tail of the distribution of \(\widetilde{S}\) well even
when the assumed distribution is only loosely connected to the true one.
As with the empirical approximation, these are also accurate estimates
of the true tail probabilities \(p_i\) provided the assumed and true
distributions order the surprisals in the same way above the threshold.
This depends on the assumed distribution rather than on the sample size.

The GPD approximation also requires a threshold. We have used
\(\beta = 0.1\) throughout. In the bivariate Gamma experiment, the false
anomaly rate at \(\alpha = 0.01\) was almost the same for \(\beta\)
between 0.05 and 0.2, but at \(\alpha = 0.001\) a threshold as low as
\(\beta = 0.2\) led to too few anomalies when the assumed distribution
was lighter-tailed than the data (see the supplementary material). So we
suggest \(\beta = 0.1\) as a default, with smaller values when \(n\) is
large and very small tail probabilities are of interest. The choice can
be checked with standard peaks-over-threshold diagnostics such as mean
residual life and parameter stability plots \citep[Ch.
4]{coles2001introduction}, which the supplementary material gives for
each of our applications. A systematic study of threshold selection for
surprisal data is a topic for future work.

In practice, the choice between the empirical and GPD approximations
depends on the sample size, \(n\), and the desired level of
significance, \(\alpha\). The empirical approximation requires no tuning
or model estimation, and provides finite-sample uniform error control,
but it can be noisy when \(n\alpha\) is small; as discussed in
Section~\ref{sec-empirical}, we suggest \(n\alpha \ge 10\) as a rough
minimum. The GPD approximation is most useful when \(n\alpha\) is small,
as it can reduce variance and allow extrapolation beyond the observed
maximum, but it requires choosing a threshold and fitting a parametric
model to the tail, which introduces additional complexity and potential
for mis-specification. For the GPD approximation, what matters is the
number of exceedances, \(n\beta\), rather than \(n\alpha\). In the
bivariate Gamma experiment, the false anomaly rate of the GPD
approximation at \(\alpha = 0.01\) was 0.013 with \(n\beta = 10\)
exceedances, but within 0.001 of the nominal rate once \(n\beta \ge 50\)
(that is, \(n \ge 500\) with \(\beta = 0.1\)), so we suggest
\(n\beta \ge 50\) as a rough minimum. When the surprisal distribution
exhibits discreteness, spikes, or other non-smooth features that one may
not wish to smooth away with a parametric tail model, the empirical
approximation may be preferable. Often it is worth computing both, as a
check on each other. While both approaches can be expected to perform
well under moderate model mis-specification, neither approach is
guaranteed if the mis-specification is severe enough that the tail
orderings of the true and fitted distributions differ substantially.

Surprisals are defined in the same way for discrete, continuous and
mixed distributions, and need no notion of distance or centre, which is
hard to define for skewed, multi-modal, or high-dimensional data. The
method is also cheap to compute: surprisal computation is linear in the
sample size, empirical tail estimation requires only sorting, and GPD
fitting is fast even for very large datasets.

Because anomalies are defined by low density rather than by distance
from a centre, the framework finds inlying anomalies between the modes
of a multimodal distribution as well as outlying anomalies in the tails.
In the bimodal simulation in Section~\ref{sec-experiments}, the
surprisal-based methods find most of the inlying anomalies while the
local outlier factor, a standard density-based method, finds about half,
and the Hampel identifier, a robust distance-based method, finds none.
In the Old Faithful application, several eruptions are flagged because
their combination of duration and waiting time falls between the two
clusters, although neither value is extreme on its own. In both cases
the Gaussian mixture was fitted to data that included the anomalies, and
the empirical and GPD approximations still worked.

Our theory assumes independent observations. With serially dependent
data, such as the French mortality series, both approximations still
estimate the marginal tail probabilities, but with less precision than
the same number of independent observations would give, and anomalies
tend to occur in clusters. Extending the finite-sample guarantee of
Lemma~\ref{lem-confident} to dependent data, and using the extremal
index to describe clusters of anomalies, are topics for future work.

The surprisals depend on how the assumed distribution is estimated. If
the data contain anomalies, a non-robust estimate can be distorted by
them, inflating the density near the anomalies and so masking them.
Robust estimation methods should therefore be used where possible, and
the surprisal framework allows any estimate of \(\widetilde{F}\),
including a robust one. For example, the French mortality application
uses a running median and MAD to define the conditional distribution for
this reason.

The surprisal framework rests on a clear definition of an anomaly as an
observation of low probability, and gives both finite-sample and
extreme-value guarantees for the resulting anomaly scores, without
requiring the model to be correct. The working model is used only to
compute the surprisals; the tail probabilities are then estimated from
the surprisals themselves, using either their empirical distribution or
a fitted GPD. So the model must rank the observations correctly in the
tail, but need not get their probabilities right. In the simulations and
in all three applications this was enough. Using models that were at
best approximations, the approach found the wars and epidemics in two
centuries of French mortality data, the batters with unusually many not
outs once their skill and era were taken into account, and the Old
Faithful eruptions whose duration and waiting time are each unremarkable
but rare in combination.

Our methods are implemented in the \texttt{weird} R package
\citep{Rweird}, which computes surprisals and estimates their tail
probabilities using both the empirical and GPD approximations.

\subsection*{Software and
reproducibility}\label{software-and-reproducibility}
\addcontentsline{toc}{subsection}{Software and reproducibility}

All results presented here can be reproduced using the data and code
available at \url{https://github.com/robjhyndman/Surprisal_Theory}. The
analysis was conducted using R version 4.6.1 \citep{R}, with the
following R packages: weird \citep{Rweird}, tsibble \citep{Rtsibble},
targets \citep{targets, targetspaper}, ggplot2
\citep{ggplot2, ggplotbook}, and other tidyverse \citep{tidyverse}
packages.

\subsection*{Funding}\label{funding}
\addcontentsline{toc}{subsection}{Funding}

Both authors are funded by the Australian Research Council through
Discovery Project DP250100702. Rob Hyndman is a member of the Australian
Research Council Industrial Transformation Training Centre in
Optimisation Technologies, Integrated Methodologies, and Applications
(OPTIMA), Project ID IC200100009.

\subsection*{Generative AI statement}\label{generative-ai-statement}
\addcontentsline{toc}{subsection}{Generative AI statement}

The authors used Claude Code (Anthropic), with the Claude Opus 5.5
model, for copyediting to improve the clarity of the text, and for
coding assistance to speed up writing the R code. The authors checked
and revised all AI-assisted text and code, and take full responsibility
for the content.

\subsection*{References}\label{references}
\addcontentsline{toc}{subsection}{References}

\renewcommand{\bibsection}{}
\bibliography{refs.bib}

@book{Bourdelais1987,
  address   = {Paris},
  author    = {Bourdelais, P and Raulot, J-Y},
  publisher = {Payot},
  title     = {Une peur bleue: Histoire du choléra en {France}, 1832--1854},
  year      = 1987
}

@book{coles2001introduction,
  address   = {London UK},
  author    = {Coles, Stuart},
  publisher = {Springer},
  series    = {Springer Series in Statistics},
  title     = {An introduction to statistical modeling of extreme values},
  volume    = {208},
  year      = {2001}
}

@book{Cover2006,
  address   = {Hoboken, NJ},
  author    = {Cover, Thomas M and Thomas, Joy A},
  edition   = {2nd},
  publisher = {Wiley},
  title     = {Elements of information theory},
  year      = {2006}
}

@article{Davies1993,
  author  = {Davies, Laurie and Gather, Ursula},
  doi     = {10.1080/01621459.1993.10476339},
  issue   = {423},
  journal = {Journal of the American Statistical Association},
  pages   = {782--792},
  title   = {The identification of multiple outliers},
  volume  = {88},
  year    = {1993}
}

@misc{french_mortality,
  author = {Dana Glei and John Wilmoth and Magali Barbieri and Jacques Vallin and France Meslé},
  note   = {Accessed 2026-04-29},
  title  = {About mortality data for {France}, total population},
  url    = {https://mortality.org/File/GetDocument/hmd.v6/FRATNP/Public/InputDB/FRATNPcom.pdf},
  year   = 2025
}

@manual{ggplot2,
  author = {Hadley Wickham and Winston Chang and Lionel Henry and Thomas Lin Pedersen and Kohske Takahashi and Claus Wilke and Kara Woo and Hiroaki Yutani and Dewey Dunnington and Teun van den Brand and {Posit, PBC}},
  note   = {R package version 4.0.3},
  title  = {{ggplot2}: Create Elegant Data Visualisations Using the Grammar of Graphics},
  url    = {https://cran.r-project.org/package=ggplot2},
  year   = {2026}
}

@book{ggplotbook,
  address   = {New York, USA},
  author    = {Hadley Wickham},
  isbn      = {978-3-319-24277-4},
  publisher = {Springer-Verlag},
  title     = {{ggplot2}:  Elegant Graphics for Data Analysis},
  url       = {https://ggplot2-book.org/},
  year      = {2016}
}

@article{Good1952-zw,
  author  = {Good, I J},
  issue   = {1},
  journal = {Journal of the Royal Statistical Society. Series B,
             Statistical methodology},
  pages   = {107--114},
  title   = {Rational decisions},
  volume  = {14},
  year    = {1952}
}

@article{Grubbs1950,
  author    = {Grubbs, F E},
  journal   = {Annals of Mathematical Statistics},
  month     = mar,
  number    = 1,
  pages     = {27--58},
  publisher = {Institute of Mathematical Statistics},
  title     = {Sample criteria for testing outlying observations},
  volume    = 21,
  year      = 1950
}

@book{Hampel1986,
  address   = {New York},
  author    = {Hampel, F R and Ronchetti, E M and Rousseeuw, Peter J and Stahel, Werner A},
  publisher = {John Wiley \& Sons},
  title     = {Robust Statistics: The Approach based on Influence Functions},
  year      = 1986
}

@article{hdoutliers,
  author  = {Wilkinson, Leland},
  doi     = {10.1109/TVCG.2017.2744685},
  journal = {IEEE Transactions on Visualization and Computer Graphics},
  number  = {1},
  pages   = {256-266},
  title   = {Visualizing Big Data Outliers Through Distributed Aggregation},
  volume  = {24},
  year    = {2018}
}

@article{HDR96,
  author  = {Rob J Hyndman},
  journal = {The American Statistician},
  number  = {2},
  pages   = {120--126},
  title   = {Computing and graphing highest density regions},
  volume  = {50},
  year    = {1996}
}

@misc{HMD,
  author = {{Human Mortality Database}},
  note   = {Accessed on 6 October 2025},
  title  = {{University of California, Berkeley (USA); Max Planck Institute for Demographic Research (Germany)}},
  url    = {www.mortality.org},
  year   = {2025}
}

@article{lookout,
  author  = {Sevvandi Kandanaarachchi and Rob J Hyndman},
  doi     = {10.1080/10618600.2021.2000425},
  issue   = 2,
  journal = {J Computational \& Graphical Statistics},
  pages   = {586-599},
  title   = {Leave-one-out kernel density estimates for outlier detection},
  volume  = 31,
  year    = 2022
}

@unpublished{lookout2,
  author = {Rob J Hyndman and Sevvandi Kandanaarachchi and Katharine Turner},
  title  = {When lookout sees crackle: Anomaly detection via kernel density estimation},
  url    = {https://robjhyndman.com/publications/lookout2.html},
  year   = 2026
}

@manual{mgcv,
  author = {Wood, S N},
  note   = {R package version 1.9-4},
  title  = {Mixed GAM Computation Vehicle with Automatic Smoothness Estimation},
  url    = {https://cran.r-project.org/package=mgcv},
  year   = 2025
}

@manual{R,
  address      = {Vienna, Austria},
  author       = {{R Core Team}},
  note         = {Version 4.6.1},
  organization = {R Foundation for Statistical Computing},
  title        = {R: A Language and Environment for Statistical Computing},
  url          = {https://www.R-project.org/},
  year         = {2026}
}

@manual{Rtsibble,
  author = {Earo Wang and Di Cook and Rob J Hyndman and Mitchell O'Hara-Wild and Tyler Smith and Wil Davis},
  note   = {R package version 1.2.0},
  title  = {{tsibble}: {Tidy Temporal Data Frames and Tools}},
  url    = {https://tsibble.tidyverts.org},
  year   = {2026}
}

@manual{Rweird,
  author = {Rob J Hyndman},
  note   = {R package version 3.1.0},
  title  = {{weird}: {Functions and Data Sets for "That's Weird: Anomaly Detection Using R" by Rob J Hyndman}},
  url    = {https://pkg.robjhyndman.com/weird/},
  year   = {2026}
}

@book{Stone2022,
  address   = {Sheffield, England},
  author    = {Stone, James V},
  publisher = {Sebtel Press},
  title     = {Information theory: a tutorial introduction},
  year      = {2022}
}

@manual{targets,
  author = {William Michael Landau},
  note   = {R package version 1.12.0},
  title  = {{targets}: Dynamic Function-Oriented `Make'-Like Declarative Pipelines},
  url    = {https://cran.r-project.org/package=targets},
  year   = {2026}
}

@article{targetspaper,
  author  = {William Michael Landau},
  journal = {Journal of Open Source Software},
  number  = {57},
  pages   = {2959},
  title   = {The targets {R} package: a dynamic Make-like function-oriented pipeline toolkit for reproducibility and high-performance computing},
  volume  = {6},
  year    = {2021}
}

@article{tidyverse,
  author  = {Hadley Wickham and Mara Averick and Jennifer Bryan and Winston Chang and Lucy D'Agostino McGowan and Romain François and Garrett Grolemund and Alex Hayes and Lionel Henry and Jim Hester and Max Kuhn and Thomas Lin Pedersen and Evan Miller and Stephan Milton Bache and Kirill Müller and Jeroen Ooms and David Robinson and Dana Paige Seidel and Vitalie Spinu and Kohske Takahashi and Davis Vaughan and Claus Wilke and Kara Woo and Hiroaki Yutani},
  journal = {Journal of Open Source Software},
  number  = {43},
  pages   = {1686},
  title   = {Welcome to the {tidyverse}},
  volume  = {4},
  year    = {2019}
}

@book{Tribus61,
  address   = {New York, USA},
  author    = {Tribus, Myron},
  publisher = {D. Van Nostrand},
  title     = {Thermodynamics and thermostatics: An introduction to energy, information and states of matter, with engineering applications},
  year      = {1961}
}

@book{Wood2017gam,
  address   = {Boca Raton, FL},
  author    = {Wood, S N},
  publisher = {CRC press},
  title     = {Generalized additive models: An introduction with {R}},
  year      = {2017}
}

@inproceedings{Breunig2000,
  author    = {Breunig, Markus M. and Kriegel, Hans-Peter and Ng, Raymond T. and Sander, J{\"o}rg},
  title     = {{LOF}: Identifying Density-Based Local Outliers},
  booktitle = {Proceedings of the 2000 ACM SIGMOD International Conference on Management of Data},
  pages     = {93--104},
  year      = {2000},
  publisher = {ACM},
  doi       = {10.1145/342009.335388}
}

@article{mclust,
  author  = {Scrucca, Luca and Fop, Michael and Murphy, T. Brendan and Raftery, Adrian E.},
  title   = {{mclust} 5: Clustering, Classification and Density Estimation Using {G}aussian Finite Mixture Models},
  journal = {The R Journal},
  volume  = {8},
  number  = {1},
  pages   = {289--317},
  year    = {2016},
  doi     = {10.32614/RJ-2016-021}
}

\newpage{}

\appendix

\section{\texorpdfstring{Proof of
Lemma~\ref{lem-confident}}{Proof of Lemma~}}\label{proof-of-lem-confident}

\begin{proof}
Define \(\mathbb{P}_n(s)=n^{-1}\sum_{i=1}^{n}1[S_i\le s]\). Fix an index
\(i\) with \(\widetilde{S}_i \geq s_{\star}\). For any \(j\), consider
the two possibilities. If \(\widetilde{S}_j \geq s_{\star}\), strict
monotonicity gives \[
\widetilde{S}_j \leq \widetilde{S}_i \quad \Longleftrightarrow \quad h\left(\widetilde{S}_j\right) \leq h\left(\widetilde{S}_i\right) \quad \Longleftrightarrow \quad S_j \leq S_i .
\] If \(\widetilde{S}_j<s_{\star}\), then both inequalities hold,
because \[
\widetilde{S}_j<s_{\star} \leq \widetilde{S}_i, \quad S_j \leq h\left(s_{\star}\right) \leq h\left(\widetilde{S}_i\right)=S_i .
\] Consequently,
\(1\left[\widetilde{S}_j \leq \widetilde{S}_i\right]=1\left[S_j \leq S_i\right]\)
for every \(j\). Summing over \(j\) yields the exact identity \[
\widehat{\mathbb{P}}_n(\widetilde{S}_i)=\mathbb{P}_n\left(S_i\right) .
\] Therefore, simultaneously for all eligible indices, \[
\left|\widehat{\mathbb{P}}_n(\widetilde{S}_i)-G\left(S_i\right)\right|=\left|\mathbb{P}_n\left(S_i\right)-G\left(S_i\right)\right| \leq \sup _{s \in \mathbb{R}}\left|\mathbb{P}_n(s)-G(s)\right| .
\] Take \(\epsilon=\sqrt{\log (2/\delta)/(2n)}\). From the
Dvoretzky-Kiefer-Wolfowitz theorem, \begin{align*}
\Pr\left\{\exists i\;:\; \;\widetilde{S}_i\ge s_\star,\; \left|\widehat{\mathbb{P}}_n(\widetilde{S}_i)-G\left(S_i\right)\right|>\epsilon\right\}&\le\Pr\left\{\sup_{s\in\mathbb{S}}|\mathbb{P}_n(s)-G(s)|>\epsilon\right\}\\&\le 2e^{-2n\epsilon^2}\\&=2e^{-2n \left(\frac{\log(2/\delta)}{2n}\right)}\\&=2e^{-\log(2/\delta)}\\&=\delta.
\end{align*} Re-arranging terms yields the stated inequality: for all
\(\widetilde{S}_i\ge s_\star\), we have
\(\widehat{\mathbb{P}}_n(\widetilde{S}_i)-\epsilon\le G(S_i)\le \widehat{\mathbb{P}}_n(\widetilde{S}_i)+\epsilon\)
with probability at least \(1-\delta\).
\end{proof}

\section{\texorpdfstring{Proof of
Theorem~\ref{thm-gpd}}{Proof of Theorem~}}\label{proof-of-thm-gpd}

We first state and prove two lemmas.

\subsection{Lemmas}\label{lemmas}

In this section, \(S\) denotes a generic random variable satisfying the
stated conditions; in Theorem~\ref{thm-gpd} it is the surprisal
\(\widetilde{S}\). Let \(X = S - \mu\), and let
\(\mathcal{X} = \text{supp}(X)\). In this appendix, redefine \(G\) as
the distribution function of \(X\), define \(\bar{G} = 1 - G\). Define
the sequence \(b_n\) by the relation \(\bar{G}(b_n) = 1/n\), for all
\(n\) large enough, and let \(a_n = 1/\{n \cdot g(b_n)\} > 0\), where
\(g\) denotes the density of \(G\). Define
\(M_n = \max\{S_1 - \mu,\dots,S_n - \mu\}\). We maintain the following
regularity conditions on \(X\).

\begin{definition}[Regularity
Condition]\protect\hypertarget{def-regular}{}\label{def-regular}

~

\begin{enumerate}
\def\labelenumi{(\roman{enumi})}
\tightlist
\item
  For some \(t_0 < x_+ := \sup\{x \in \mathcal{X}:G(x) < 1\} = \infty\),
  \(G\) is absolutely continuous on \((t_0,x_+)\) with a strictly
  positive and continuously differentiable density \(g\). (ii)
  \(\bar{G}(t) \rightarrow 0\) as \(t \rightarrow x_+\). (iii) The
  reciprocal hazard function \(r(t) = \bar{G}(t)/g(t)\) satisfies \[
    \lim_{t \uparrow x_+} \frac{d}{dt}r(t) = \xi \in \mathbb{R}.
  \]
\end{enumerate}

\end{definition}

Define \[
  G_\xi(x) := \begin{cases}
                \exp\{-(1 + \xi x)^{-1/\xi}\},\quad\xi \ne 0,\quad1 + \xi x > 0, \\
                \exp\{-e^{-x}\},\quad\xi = 0
              \end{cases}.
\]

\begin{lemma}[]\protect\hypertarget{lem-GEVconv}{}\label{lem-GEVconv}

Let Assumption~\ref{def-regular} be satisfied. If
Assumption~\ref{def-bounded} is satisfied then, for every
\(x \in \mathbb{R}\), \[
  \Pr\left( \frac{M_n - b_n}{a_n} \le x \right) \longrightarrow G_0(x).
\] If instead Assumption~\ref{def-poly} is satisfied, then \[
  \Pr\left( \frac{M_n - b_n}{a_n} \le x \right) \longrightarrow G_\xi(x), \quad 0 \le \xi \le 1/q.
\]

\end{lemma}

\begin{proof}
The proof proceeds in four steps.

\textbf{Step 1: Upper tail threshold.} For \(T>\max(t_0,0)\) and
\(t\ge T\), the fundamental theorem of calculus gives \[
  \frac{r(t)}{t}-\xi
  =\frac{r(T)-\xi T}{t}
   +\frac1t\int_T^t\{r'(u)-\xi\}\,du.
\] Given \(\varepsilon>0\), choose \(T\) so that
\(|r'(u)-\xi|\le\varepsilon\) for all \(u\ge T\). The absolute value of
the integral term is then at most \(\varepsilon(t-T)/t\), and the first
term tends to zero. Since \(\varepsilon\) is arbitrary, \begin{equation}
  \frac{r(t)}{t}\longrightarrow\xi.
  \label{eq:scale-order}
\end{equation} As \(r(t)>0\) and \(t>0\) for all sufficiently large
\(t\), this implies \(\xi\ge0\).

Fix \(x\in\mathbb R\) such that \(1+\xi x>0\); when \(\xi=0\), this
permits every real \(x\). Define \[
  J_x=[\min(0,x),\max(0,x)],\qquad
  c_x=\min_{v\in J_x}(1+\xi v)
     =\min(1,1+\xi x)>0.
\] By \eqref{eq:scale-order}, \[
  \sup_{v\in J_x}
  \left|
    \frac{t+r(t)v}{t}-(1+\xi v)
  \right|
  \le |x|\left|\frac{r(t)}{t}-\xi\right|
  \longrightarrow0.
\] Consequently, for all sufficiently large \(t\), \begin{equation}
  \inf_{v\in J_x}\{t+r(t)v\}
  \ge\frac{c_x}{2}t
  \longrightarrow\infty.
  \label{eq:moving-interval}
\end{equation} Thus the entire interval between \(t\) and \(t+r(t)x\)
eventually lies in \((t_0,\infty)\), and its smallest point tends to
infinity. This permits the use of Assumption~\ref{def-regular}(iii)
uniformly along that interval, including when \(x<0\).

\textbf{Step 2: Tail ratio.} Set \[
  \eta_t=\sup_{w\in J_x}|r'(t+r(t)w)-\xi|.
\] Assumption~\ref{def-regular}(iii) and equation
\eqref{eq:moving-interval} imply \(\eta_t\to0\). For every \(v\in J_x\),
the fundamental theorem of calculus gives \[
  \frac{r(t+r(t)v)}{r(t)}
  =1+\int_0^v r'(t+r(t)w)\,dw.
\] All integrals with a negative upper limit are interpreted as oriented
integrals. Hence \begin{equation}
  \sup_{v\in J_x}
  \left|
    \frac{r(t+r(t)v)}{r(t)}-(1+\xi v)
  \right|
  \le |x|\eta_t
  \longrightarrow0.
  \label{eq:local-scale}
\end{equation} Since \(1+\xi v\ge c_x>0\) on \(J_x\), the ratios on the
left are eventually bounded below by \(c_x/2\). Taking reciprocals in
\eqref{eq:local-scale} is therefore valid uniformly on \(J_x\), yielding
\begin{equation}
  \sup_{v\in J_x}
  \left|
    \frac{r(t)}{r(t+r(t)v)}-\frac{1}{1+\xi v}
  \right|
  \longrightarrow0.
  \label{eq:reciprocal-scale}
\end{equation}

On \((t_0,\infty)\), positivity of \(\overline G\) and \(g\) gives \[
  \frac{d}{dt}\log\overline G(t)
  =-\frac{g(t)}{\overline G(t)}
  =-\frac{1}{r(t)}.
\] Thus, for any \(u,v>t_0\), \begin{equation}
  \log\frac{\overline G(v)}{\overline G(u)}
  =-\int_u^v\frac{dw}{r(w)}.
  \label{eq:hazard-integral}
\end{equation} We now evaluate this integral with endpoints \(t\) and
\(t+r(t)x\). These endpoints and the interval between them are valid by
\eqref{eq:moving-interval}. Substituting \(u=t+r(t)v\) and using the
uniform convergence in \eqref{eq:reciprocal-scale} on the fixed interval
\(J_x\) gives \[
  \begin{aligned}
    \log\frac{\overline G(t+r(t)x)}{\overline G(t)}
    &=-\int_t^{t+r(t)x}\frac{du}{r(u)}\\
    &=-\int_0^x\frac{r(t)}{r(t+r(t)v)}\,dv\\
    &\longrightarrow-\int_0^x\frac{dv}{1+\xi v}\\
    &=
    \begin{cases}
      -\xi^{-1}\log(1+\xi x),&\xi>0,\\
      -x,&\xi=0.
    \end{cases}
  \end{aligned}
\] In particular, the absolute difference between the two integrals over
\([0,x]\) is bounded by \(|x|\) times the supremum in
\eqref{eq:reciprocal-scale}. Exponentiating proves \begin{equation}
  \frac{\overline G(t+r(t)x)}{\overline G(t)}
  \longrightarrow
  \tau_\xi(x):=
  \begin{cases}
    (1+\xi x)^{-1/\xi},&\xi>0,\\
    e^{-x},&\xi=0,
  \end{cases}
  \qquad 1+\xi x>0.
  \label{eq:tail-ratio}
\end{equation}

\textbf{Step 3: Convert the tail ratio.} The function \(\overline G\) is
continuous and strictly decreasing on \((t_0,\infty)\), and tends to
zero as \(t\to\infty\). It is strictly positive at every finite \(t\),
because \(x_+=\infty\). Therefore \(\overline G(b_n)=1/n\) has a unique
solution in \((t_0,\infty)\) for all sufficiently large \(n\), and
\(b_n\to\infty\). With \(a_n=r(b_n)\), equation \eqref{eq:tail-ratio}
gives \begin{equation}
  n\overline G(b_n+a_nx)
  =\frac{\overline G(b_n+r(b_n)x)}{\overline G(b_n)}
  \longrightarrow\tau_\xi(x),
  \qquad 1+\xi x>0.
  \label{eq:normalized-tail}
\end{equation} Write \(q_n=\overline G(b_n+a_nx)\). Then \(q_n\to0\) and
\(nq_n\to\tau_\xi(x)\), so \[
  n\log(1-q_n)=-nq_n\{1+o(1)\}
  \longrightarrow-\tau_\xi(x).
\] Independence now yields \[
  \Pr\left(\frac{M_n-b_n}{a_n}\le x\right)
  =G(b_n+a_nx)^n
  =(1-q_n)^n
  \longrightarrow\exp\{-\tau_\xi(x)\}
  =G_\xi(x).
\]

If \(\xi = 0\), this proves the first part of the result for every
\(x \in \mathbb R\). If \(\xi > 0\) and \(x \le -1/\xi\), then for every
\(z > -1/\xi\), monotonicity gives \[
  0 \le \limsup_{n \to \infty} \Pr\left( \frac{M_n - b_n}{a_n} \le x \right) \le G_\xi(z).
\] Letting \(z \downarrow -1/\xi\) makes the right-hand side tend to
zero. Thus the result also holds at and below the lower support
boundary.

\textbf{Step 4: Moment Implications.} For every \(q>0\), Tonelli's
theorem gives the tail-integral identity \begin{equation}
  \mathbb E[X_+^q]
  =q\int_0^\infty t^{q-1}\overline G(t)\,dt,
  \label{eq:tail-moments}
\end{equation} with equality also when the expectation is infinite, and
where \(X_+=\max(X,0)\).

Suppose first that \(\xi>0\). By \eqref{eq:scale-order}, for any
\(\varepsilon\in(0,\xi)\) there is a \(T>\max(t_0,0)\) such that \[
  (\xi-\varepsilon)t\le r(t)\le(\xi+\varepsilon)t,
  \qquad t\ge T.
\] Taking reciprocals and integrating gives \[
  \frac{\log(t/T)}{\xi+\varepsilon}
  \le \int_T^t\frac{du}{r(u)}
  \le \frac{\log(t/T)}{\xi-\varepsilon}.
\] Using \eqref{eq:hazard-integral}, we obtain \begin{equation}
  \overline G(T)(t/T)^{-1/(\xi-\varepsilon)}
  \le\overline G(t)
  \le\overline G(T)(t/T)^{-1/(\xi+\varepsilon)},
  \qquad t\ge T.
  \label{eq:power-bounds}
\end{equation} If \(q<1/\xi\), choose \(\varepsilon>0\) sufficiently
small that \(q<1/(\xi+\varepsilon)\). The upper bound in
\eqref{eq:power-bounds} then makes the tail integral in
\eqref{eq:tail-moments} finite. If \(q>1/\xi\), choose \(\varepsilon>0\)
sufficiently small that \(q>1/(\xi-\varepsilon)\). The lower bound makes
that integral infinite. The integral over \([0,T]\) is always at most
\(T^q\). Consequently, \begin{equation}
  \begin{cases}
    \mathbb E[X_+^q]<\infty,&0<q<1/\xi,\\
    \mathbb E[X_+^q]=\infty,&q>1/\xi.
  \end{cases}
  \label{eq:moment-index}
\end{equation} These conclusions impose no restriction on whether the
moment of order exactly \(1/\xi\) is finite.

Suppose next that \(\xi=0\). Equation \eqref{eq:scale-order} gives
\(r(t)=o(t)\). For any \(A>0\), choose \(T>\max(t_0,0)\) sufficiently
large that \(r(t)\le t/A\) for \(t\ge T\). Equation
\eqref{eq:hazard-integral} then implies \[
  \overline G(t)
  =\overline G(T)
    \exp\left\{-\int_T^t\frac{du}{r(u)}\right\}
  \le\overline G(T)(t/T)^{-A},\qquad t\ge T.
\] For each \(q>0\), taking \(A>q\) in \eqref{eq:tail-moments} proves
\begin{equation}
  \mathbb E[X_+^q]<\infty
  \qquad\text{for every }q>0
  \quad\text{when }\xi=0.
  \label{eq:gumbel-moments}
\end{equation}

Under Assumption~\ref{def-bounded}, \[
  \mathbb E[e^{\lambda X_+}]
  \le 1+\mathbb E[e^{\lambda X}]<\infty.
\] For every \(q>0\), the function \(u^q e^{-\lambda u}\) is bounded on
\([0,\infty)\), so every positive-part moment is finite. Equation
\eqref{eq:moment-index} therefore excludes \(\xi>0\). Together with
\(\xi\ge0\), this gives \(\xi=0\).

Under Assumption~\ref{def-poly}, if \(\xi>0\), equation
\eqref{eq:moment-index} requires \(q\le1/\xi\), or \(\xi\le1/q\). The
same upper bound holds when \(\xi=0\). If additionally
\(\mathbb E[X_+^r]=\infty\) for some \(r>q\), equation
\eqref{eq:gumbel-moments} excludes \(\xi=0\). Equation
\eqref{eq:moment-index} then requires \(r\ge1/\xi\), which gives
\(\xi\ge1/r\).

Let \(p_*:=\sup\{q>0:\E\max(X,0)^q<\infty\}\). Equation
\eqref{eq:gumbel-moments} gives \(p_*=\infty\) when \(\xi=0\), whereas
\eqref{eq:moment-index} gives \(p_*=1/\xi>0\) when \(\xi>0\), regardless
of the finiteness of the boundary moment. Thus a finite moment index
satisfies \(\xi=1/p_*\).
\end{proof}

Recall, for \(y \ge 0\), \[
  H_{\xi}(y) := \begin{cases}
                  1 - (1 + \xi y)^{-1 / \xi}, & \xi \neq 0, \quad 1 + \xi y > 0 \\
                  1 - e^{-y},                 & \xi = 0,
                \end{cases}
\] and \(H_{\xi}(y) = 1\) for \(y \geq -1 / \xi\) when \(\xi < 0\).
Recall that \(X = S - \mu\). Theorem~\ref{thm-gpd} in the main text
follows from the following Lemma.

\begin{lemma}[]\protect\hypertarget{lem-GPDconv}{}\label{lem-GPDconv}

Let Assumption~\ref{def-regular} be satisfied. If
Assumption~\ref{def-bounded} is satisfied, then \[
  \Pr\left( \frac{X - b_n}{a_n} > y\big{|}X > b_n \right) \longrightarrow e^{-y}.
\] If Assumption~\ref{def-poly} is satisfied, then \[
  \Pr\left( \frac{X - b_n}{a_n} > y\big{|}X > b_n \right) \longrightarrow (1 + \xi y)^{-1/\xi},\quad (1 + \xi y) > 0.
\]

\end{lemma}

\begin{proof}
We prove a general form of the result and then specialize the result to
the assumed moment conditions.

Recall \(G(s)=\Pr(X\le s)\). From Lemma~\ref{lem-GEVconv} and
independence, we have that \(G(b_n+a_nx)^{n}\rightarrow G_\xi(x)\) for
any \(x\) such that \(1+\xi x>0\). Let \(\bar{G}=1-G\). However,
\begin{flalign}
\label{eq:conv1}
n\bar{G}(b_n+a_n x)\rightarrow -\log G_\xi(x),
\end{flalign} since \(\log(1-t)\sim -t\) for \(t\downarrow0\) and
sufficiently small, and since
\(\bar{G}(b_n+a_nx)=\Pr(X>b_n+a_n x)\rightarrow0\) as
\(n\rightarrow\infty\). Let \(\tau_\xi(x):=-\log G_\xi(x)\).

For given \(u<x_+\), but sufficiently close, let
\(m(u)=\lfloor 1/\bar{G}(u)\rfloor\). The result of
Lemma~\ref{lem-GEVconv} excludes the case where an atom at a finite
upper endpoint exists; hence, as \(u\uparrow x_+\),
\begin{flalign}\label{eq:gpd1}
 \bar{G}(u)\rightarrow0,\quad m(u)\rightarrow\infty,\quad m(u)\bar{G}(u)\rightarrow1.
\end{flalign}

We now show that, \[
z(u)=\frac{u-b_{m(u)}}{a_{m(u)}}\xrightarrow{u\rightarrow x_+}0.
\] For any \(\varepsilon>0\), and sufficiently small, \eqref{eq:conv1}
implies that \[
m(u)\bar{G}\{b_{m(u)}-a_{m(u)}\varepsilon\}\rightarrow\tau_\xi(-\varepsilon);\quad  m(u)\bar{G}\{b_{m(u)}+a_{m(u)}\varepsilon\}\rightarrow\tau_\xi(\varepsilon).
\] From \eqref{eq:gpd1}, and monotonicity of \(\bar{G}\), we have, for
sufficiently large \(u\), \[
b_{m(u)}-a_{m(u)}\varepsilon<u<b_{m(u)}+a_{m(u)}\varepsilon,
\] and so we have that \(z(u)\rightarrow0\).

Fix \(x\ge0\) with \(1+\xi x>0\), and choose \(\varepsilon>0\) small
enough so that eventually \[
x-\varepsilon\le x+z(u)\le x+\varepsilon
\] Now, since \(a_{m(u)}>0\), for \(u=b_{m(u)}+a_{m(u)}z(u)\), we have
\[
b_{m(u)}+a_{m(u)}(x-\varepsilon)\le u+a_{m(u)}x\le b_{m(u)}+a_{m(u)}(x+\varepsilon).
\] However, \(\bar{G}\) is decreasing so that \[
m(u)\bar{G}\{b_{m(u)}+a_{m(u)}(x+\varepsilon)\}\le m(u)\bar{G}\{ u+a_{m(u)}x\}\le m(u)\bar{G}\{b_{m(u)}+a_{m(u)}(x-\varepsilon)\}
\] From \eqref{eq:conv1}, the outer-term in the above equation converge
to \(\tau_\xi(x+\varepsilon)\) and \(\tau_\xi(x-\varepsilon)\),
respectively. Therefore, since \(\tau_\xi(x)\) is continuous, taking
\(\varepsilon\rightarrow0\) implies that \[
 m(u)\bar{G}\{ u+a_{m(u)}x\}=m(u)\bar{G}\left[b_{m(u)}+a_{m(u)}\{z(u)+x\}\right]\rightarrow\tau_\xi(x).
\] Consequently, for \(\omega(u)=a_{m(u)}\), \[
\Pr\left\{\frac{X-u}{\omega(u)}>x\mid X>u\right\}=\frac{m(u)\bar{G}\{u+a_{m(u)}x\}}{m(u)\bar{G}(u)}\rightarrow \tau_\xi(x)=1-H_\xi(x).
\]

Given the above convergence, a point-wise version of the stated result
now follows by the results in Lemma~\ref{lem-GEVconv}.

The above convergence can be extended to uniform convergence directly.
Define, \[
  K_u(y) = \operatorname{Pr}\left( \left. \frac{X - u}{\omega(u)} \leq y \right\rvert\, X > u \right) .
\] Clearly, the earlier argument implies, for every \(y \ge 0\),
\(K_u(y) \longrightarrow H_{\xi}(y)\). However, \(K_u\) and \(H_\xi\)
are distribution functions and \(H_{\xi}\) is continuous and proper, we
have \[
  \sup_{y \geq 0}\left| K_u(y) - H_{\xi}(y) \right| \longrightarrow 0
\] as claimed.
\end{proof}

\end{document}